\newif\ifels
\elsfalse
\ifels
	\documentclass[1p, sort&compress]{elsarticle}
	\newproof{proof}{Proof}
	\newtheorem{theorem}{Theorem}
	\newtheorem{proposition}{Proposition}
	\journal{Applied and Computational Harmonic Analysis}
	\usepackage{lineno}
	\usepackage{color}
\else
	\documentclass[10pt, onecolumn, twoside]{article}

\usepackage{titlesec, color, amsthm}
\definecolor{blue1}{rgb}{0,0,0}
\renewcommand\thesection{\arabic{section}}
\renewcommand{\theequation}{\arabic{equation}}

\titleformat{\section}[hang]{\color{blue1}\large\bfseries\sffamily}{\thesection}{0mm}{. }[]
\titleformat{\subsection}[hang] {\color{blue1}\bfseries\sffamily}{\thesubsection}{0em}{. }[]
\titleformat{\subsubsection}[hang] {\color{blue1}\sffamily}{\thesubsubsection}{0em}{. }[]
\titlespacing*{\section}{1em}{3.5ex plus .2ex minus .2ex}{1ex plus .2ex}
\titlespacing*{\subsection}{0em}{3ex plus .2ex minus .2ex}{1ex plus .2ex}
\titlespacing*{\subsubsection}{0em}{3ex plus .2ex minus .2ex}{1ex plus .2ex}

\renewenvironment{abstract}{{\color{blue1}\small\bfseries Abstract.}\footnotesize}{\par \vskip .1in}
\makeatletter
\def\@setauthors{
\begingroup 
\def \thanks{\protect\thanks@warning}
\trivlist \centering\footnotesize \@topsep30\p@\relax \advance\@topsep by -\baselineskip
\item\relax \author@andify \authors \def\\{\protect\linebreak} {\color{blue1}\large\authors} \endtrivlist \endgroup}
\def\@settitle{\centering{\color{blue1} \Large \bfseries \bfseries \@title \par}}
\makeatother
\usepackage[normal, small, up, textfont=it]{caption}
	\usepackage{cite} 
	\newtheorem{theorem}{Theorem}[section]
	\newtheorem{proposition}{Proposition}[section]
	\usepackage{lineno}
\fi
\usepackage{amsmath} \usepackage{algorithm, algpseudocode}
\usepackage{amssymb, bm, graphicx, hyperref, enumerate, rotating}
\usepackage{verbatim} 
\usepackage{subfigure}
\usepackage{epstopdf}

\renewcommand{\leq}{\ensuremath{\leqslant}}
\renewcommand{\geq}{\ensuremath{\geqslant}}

\newtheorem{definition}[theorem]{Definition}

\algnewcommand\algorithmicinput{\textbf{Input:}}
\algnewcommand\Input{\item[\algorithmicinput]}

\algnewcommand\algorithmicoutput{\textbf{Output:}}
\algnewcommand\Output{\item[\algorithmicoutput]}
\newcommand{\longtitle}{{Design of graph filters and filterbanks}}
\newcommand{\NTlong}{{Nicolas Tremblay}}
\newcommand{\PBlong}{{Pierre Borgnat}}
\newcommand{\PGlong}{{Paulo Gon\c{c}alves}}
\newcommand{\NTshort}{{N.~Tremblay}}
\newcommand{\PBshort}{{P.~Borgnat}}
\newcommand{\PGshort}{{P.~Gon\c{c}alves}}
\newcommand{\ENSLPB}{{Universit\'e Lyon, ENS de Lyon, Univ Lyon 1, CNRS, Laboratoire de Physique $\&$ IXXI, Lyon, France}}
\newcommand{\ENSLPG}{{Universit\'e Lyon, ENS de Lyon, Univ Lyon 1, CNRS, Inria, LIP $\&$ IXXI, Lyon, France}}
\newcommand{\GIPSANT}{{Univ. Grenoble Alpes, CNRS, GIPSA-lab, Grenoble, France}}
\newcommand{\ANR}{{This work was supported by the ANR grant \textsc{Graphsip} ANR-14-CE27-0001-02 and ANR-14-CE27-0001-03.}}
\graphicspath{{img/}}

\ifels
	\title{\longtitle\tnoteref{t1}}
	\tnotetext[t1]{\ANR.}
	\author[gipsant]{\NTlong}
	\author[enslpb]{\PBlong}
	\author[enslpg]{\PGlong}
	\address[gipsant]{\GIPSANT}
	\address[enslpb]{\ENSLPB}
	\address[enslpg]{\ENSLPG}
\else
	\title{\longtitle}
	\author{\NTlong\thanks{\NTshort\ is with \GIPSANT}, \PGlong\thanks{\PGshort\ is with \ENSLPG}, \PBlong\thanks{\PBshort\ is with \ENSLPB.}}
\fi

\begin{document}

\ifels










\begin{abstract}
Basic operations in graph signal processing consist in processing signals indexed on graphs
either by filtering them, to extract specific part out of them, or by changing their domain of representation,
using some transformation or dictionary more adapted to represent the information contained in them.
The aim of this chapter is to review general concepts for the introduction of filters and representations
of graph signals. 
We first begin by recalling the general framework to achieve that, which put the emphasis on
introducing some spectral domain that is relevant for graph signals to define a Graph Fourier Transform. 
We show how to introduce a notion of frequency analysis for graph signals by looking at their variations.
Then, we move to the introduction of graph filters, that are defined like the classical equivalent
for 1D signals or 2D images, as linear systems which operate on each frequency of a signal.
Some examples of filters and of their implementations are given. 
Finally, as alternate representations of graph signals, we focus on multiscale transforms that 
are defined from filters. 
Continuous multiscale transforms such as spectral wavelets on graphs are reviewed, as well as the versatile approaches
of filterbanks on graphs. 
Several variants of graph filterbanks are discussed, for structured as well as arbitrary graphs, with a focus on the central
point of the choice of the decimation or aggregation operators.
\end{abstract}


	\maketitle
	\tableofcontents
	\section{Graph Fourier Transform and Frequencies}
\label{chapfilt:sec1}

\subsection{Introduction}
\label{chapfilt:ssec11}
Graph Signal Processing (GSP) has been introduced in the recent past using at least two complementary
formalisms: on the one hand, the discrete signal processing on graphs \cite{sandryhaila_TSP2013} 
which emphasizes the adjacency matrix as a shift operator on graph and develops an equivalent of Discrete Signal Processing (DSP) for
signals on graphs; and on the other hand, the approaches rooted in graph spectral analysis, which
rely on the spectral properties of a Laplacian matrix (or operator) on a graph \cite{shuman_SPMAG2013,hammond_ACHA2011,narang_TSP2012}.
Both approaches yield an analogy for harmonic analysis of graph signals through the definition of
a Graph Fourier Transform as being the projection of a signal in the spectral domain of the chosen matrix or operator. 
While the technical details vary, and some interpretations in the vertex domain may differ, the fundamental objective of both approaches is to decompose a signal onto components of different frequencies and 
to design filters that can extract or modify parts of a graph signal according to these frequencies, 
e.g. providing notions of low-pass, band-pass, or high-pass filters for graph signals. 
In this section, both approaches (along with other variations) are seen as specific instances of the general guideline for defining Graph Fourier Transform and its associated frequency analysis. 

\paragraph{Notations.}
Vectors are written in bold with small letters, and matrices in bold and capital letters.
Let $\mathcal{G}=(\mathcal{V},\mathcal{E},\mathbf{A})$ be a  graph with $\mathcal{V}$ the set 
of $N$ nodes, $\mathcal{E}$ the set of edges, and $\mathbf{A}$ the weighted adjacency matrix in $\mathbb{R}^{N\times N}$.
If $\mathbf{A}_{ij}=0$, there is no connection from node $i$ to node $j$, otherwise, $\mathbf{A}_{ij}$ is the weight 
of the edge starting from $i$ and pointing to $j$\footnote{In the literature, the converse convention is sometimes chosen (e.g. in \cite{sandryhaila_TSP2013}), hence the $\mathbf{A}^\top$ occasionally appearing in this chapter.}. 
If an undirected edge exists between $i$ and $j$, then $\mathbf{A}_{ij}=\mathbf{A}_{ji}$. 
We restrict ourselves to adjacency matrices with positive or null entries: $\mathbf{A}_{ij}\geq0$. Finally, the symbol $\mathbf{I}$ denotes the identity matrix (its dimension should be clear with the context), and $\bm{\delta}_i$ is a vector whose $i$-th entry is equal to $1$ while all other entries are equal to $0$. 



\subsection{Graph Fourier Transform}
\label{chapfilt:ssec12}

\begin{definition}[graph signal]
A graph signal is a vector $\mathbf{x} \in \mathbb{R}^N$ whose component $x_i$ is considered to be defined on vertex $i$. 
\end{definition}

A Graph Fourier Transform (GFT) is defined via a choice of reference operator admitting a spectral decomposition. Representing a graph signal in this spectral domain is interpreted as a GFT. 
We review the standard properties and the various definitions proposed in the literature.

Consider  a matrix $\mathbf{R}\in\mathbb{R}^{N\times N}$. To be admissible as a reference operator for the graph, it is often required that for any pair of nodes $i\neq j$, $\mathbf{R}_{ij}$ and $\mathbf{R}_{j\,i}$ are equal to zero if $i$ and $j$ are not connected, as this will help for efficient implementations of filters (see in \ref{ssec:chapfilt:ImplementationGFilt}). 
We assume furthermore that $\mathbf{R}$ is diagonalizable in $\mathbb{C}$. In fact, if $\mathbf{R}$ is not diagonalizable,  one needs to consider Jordan's decomposition, which is out-of-scope of this chapter. We refer the reader to \cite{sandryhaila_TSP2013} for technical details on how to handle this
case. Nevertheless, in practice, we claim that it often suffices to consider only diagonalizable operators, since: {\em i)}~diagonalizable matrices in $\mathbb{C}$ are dense in the space of matrices; and {\em ii)}~graphs under consideration are generally measured (should they model social, sensor or biological networks,...) with some inherent noise. So, if one ends up unluckily with a non-diagonalizable matrix, a small perturbation within the noise level will make it diagonalizable -- provided the graphs have no specific regularities that are to be kept.
Still, we may assume with only a small loss of generality that the reference operator $\mathbf{R}$ has a spectral decomposition:
\begin{equation}
\label{eq:eigendecomp}
\mathbf{R} = \mathbf{U} \mathbf{\Lambda}  \mathbf{U}^{-1}, 
\end{equation}
with $\mathbf{U}$ and $\mathbf{\Lambda}$ in $\mathbb{C}^{N\times N}$. 
The columns of $\mathbf{U}$, denoted $\mathbf{u}_k$,  are the right eigenvectors of $\mathbf{R}$, while the lines of $\mathbf{U}^{-1}$, denoted $\mathbf{v}_k^\top$, are its left eigenvectors. $\mathbf{\Lambda}$ is the diagonal matrix of the eigenvalues $\lambda_k$. 
The GFT is defined as the transformation of a graph signal from the canonical ``node'' basis to its representation in the eigenvector basis:
\begin{definition}[Graph Fourier Transform]
For a given diagonalizable reference operator $\mathbf{R} =  \mathbf{U} \mathbf{\Lambda}  \mathbf{U}^{-1}$
acting on a graph $\mathcal{G}$, the GFT of a graph signal $\mathbf{x} \in \mathbb{R}^N$
is:
\begin{equation}
\label{eq:GFT1}
\mathbb{F}_{\mathcal{G}}\: \mathbf{x} \stackrel{.}{=} \hat{\mathbf{x}} \stackrel{.}{=} \mathbf{U}^{-1} \mathbf{x}.
\end{equation}
\end{definition}
The GFT's coefficients are simply the projections on the left eigenvectors of $\mathbf{R}$:
\begin{equation}
\label{eq:GFT2}
\forall k~~~~~\left(\mathbb{F}_{\mathcal{G}}\: \mathbf{x}\right)_k = \hat{\mathbf{x}}_k = \mathbf{v}_k^T \mathbf{x}.
\end{equation}
Moreover, the GFT is invertible: $\mathbf{U} \:\hat{\mathbf{x}}= \mathbf{U}\mathbf{U}^{-1}\mathbf{x} =\mathbf{x}$. While, in general, the complex Fourier modes $\mathbf{u}_k$ are not orthogonal to each other, when  $\mathbf{R}$ is symmetric, the following additional properties hold true. 
\\

\noindent\textbf{The special case of symmetric reference operators.} If in addition to be real, $\mathbf{R}$ is also symmetric, then $\mathbf{U}$ and $\mathbf{\Lambda}$ are real matrices, and $\mathbf{U}$ may be found orthonormal, that is: $\mathbf{U}^{-1}=\mathbf{U}^\top$. In this case, $\mathbf{v}_k=\mathbf{u}_k$, the GFT of $\mathbf{x}$ is  simply
$\hat{\mathbf{x}} = \mathbf{U}^{\top} \mathbf{x}$ with coefficients $\hat{\mathbf{x}}_k = \mathbf{u}_k^\top \mathbf{x}$, and the Parseval relation holds: $||\hat{\mathbf{x}} ||_2 =  ||{\mathbf{x}} ||_2$. Hereafter, when a symmetric operator $\mathbf{R}$ is encountered, one should have these properties in mind.  
\\

Finally, the interpretation of the graph Fourier modes $\mathbf{u}_k$ in terms of oscillations and frequencies
will be the scope of Section~\ref{chapfilt:ssec13}. In the following, we list possible choices of reference operators, all diagonalizable with different $\mathbf{U}$ and $\mathbf{\Lambda}$, thus all defining different possible GFTs. 

\paragraph{GFT for undirected graphs}
Undirected graphs are characterized by symmetric adjacency matrices: $\forall (i,j)~~\mathbf{A}_{ij}=\mathbf{A}_{j\,i}$. This does not necessarily mean that $\mathbf{R}$ has to be chosen symmetric as well, as we will see with the example $\mathbf{R}=\mathbf{L_{rw}}$ below.
The following  choices of $\mathbf{R}$ are the most common in the undirected case. 
\\

\noindent\textbf{The combinatorial Laplacian [symmetric].} 
The first choice for $\mathbf{R}$, advocated in \cite{shuman_SPMAG2013,hammond_ACHA2011,narang_TSP2012}  is to
use the graph's combinatorial Laplacian, having properties studied in \cite{chung_book1997}. It is defined as $\mathbf{L}=\mathbf{D}-\mathbf{A}$ 
where $\mathbf{D}$ is the diagonal matrix of nodes' strengths, defined as $\mathbf{D}_{ii}= \mathbf{d}_i = \sum_{j} \mathbf{A}_{ij}$.
If the adjacency matrix is binary (i.e., unweighted), this strength reduces to the degree of each node. 
The advantages of using $\mathbf{L}$ are twofold.
{\em i)}~It is an intuitive manner to define the GFT: $\mathbf{L}$ being the discretized version of the continuous Laplacian operator which admits the Fourier modes as eigenmodes, it is fair to use by analogy the eigenvectors of $\mathbf{L}$ as graph Fourier modes. Moreover, this choice is associated to a complete theory of vector calculus (e.g., gradients) for graph signals~\cite{grady2010discrete} that is useful to solve PDE's on graphs. {\em  ii)}~$\mathbf{L}$ has well known mathematical properties \cite{chung_book1997}, giving
ways to characterize the graph or functions and processes on the graph (see also \cite{BarratBook}). 
Most prominently, it is semi-definite positive (SDP) and its eigenvalues, being all positive or null, will
serve in the following to bind eigenvectors with a notion of frequency. 
\\

\noindent\textbf{The normalized Laplacian [symmetric].}  A second choice for $\mathbf{R}$ is the normalized Laplacian 
$\mathbf{L_n}=\mathbf{D}^{-\frac{1}{2}}\mathbf{L}\mathbf{D}^{-\frac{1}{2}} =\mathbf{I}-\mathbf{D}^{-\frac{1}{2}}\mathbf{A}\mathbf{D}^{-\frac{1}{2}}$. 
An interesting property of this choice of Laplacian is that all its eigenvalues lie between 0 and 2~\cite{chung_book1997}.
\\

\noindent\textbf{The adjacency matrix, or deformed Laplacian [symmetric].} Another choice for $\mathbf{R}$ is the adjacency matrix\footnote{We use $\mathbf{A}^\top$ instead of $\mathbf{A}$ for consistency purposes with~\cite{sandryhaila_TSP2013}, whose convention for directed edges in the adjacency matrix is converse to ours: what they call $\mathbf{A}$ is what we call $\mathbf{A}^\top$. Without any influence in the undirected case, it has an impact in the directed case.} $\mathbf{A}^\top$, as advocated in~\cite{sandryhaila_TSP2013}. One readily sees that the eigenbasis $\mathbf{U}$ of $\mathbf{A}^\top$ and the eigenbasis of the deformed Laplacian $\mathbf{L_d} = \mathbf{I} - \frac{\mathbf{A}^\top}{||\mathbf{A}||_2}$, where $||.||_2$ is the operator 2-norm, are the same. Therefore, the corresponding GFTs are equivalent and, for consistency in the presentation, we will use $\mathbf{R}=\mathbf{L_d}$ here. 
\\

\noindent\textbf{The random walk Laplacian [not symmetric].} The random walk Laplacian is yet another Laplacian reading: $\mathbf{L_{rw}}=\mathbf{D}^{-1}\mathbf{L} =\mathbf{I}-\mathbf{D}^{-1}\mathbf{A}$, where $\mathbf{D}^{-1}\mathbf{A}$ 
serves also to describe a uniform random walk on the graph.  
Even though $\mathbf{L_{rw}}$ is not symmetric, we know it diagonalizes in $\mathbb{R}$. In fact, if $\mathbf{u}_k$ is an eigenvector of $\mathbf{L_n}$ with eigenvalue $\lambda_k$, then $\mathbf{D}^{-\frac{1}{2}}\mathbf{u}_k$ is an eigenvector of $\mathbf{L_{rw}}$ with same eigenvalue. Thus,  $\mathbf{L_{rw}}$ has the same eigenvalues as $\mathbf{L_{n}}$, and its Fourier basis $\mathbf{U}$ is real but not orthonormal.  
\\

\noindent\textbf{Other possible definitions of the reference operator.} For instance, the consensus operator (of the form $ \mathbf{I} - \sigma \mathbf{L}$  with some suitable $\sigma$) \cite{Kar_SPM2013}, 
a geometric Laplacian~\cite{Rabiei_TMI17}, or some other deformed Laplacian one may think of, are valid alternatives. 
\\

\noindent All these operators imply a different spectral domain (different $\mathbf{U}$ and $\mathbf{\Lambda}$) and, provided one has a nice frequency interpretation (which is the object of Section~\ref{chapfilt:ssec13}), they all define possible GFTs. In the graph signal processing literature, the first three operators ($\mathbf{L}$, $\mathbf{L_n}$ and $\mathbf{L_d}$) are the most widely used.

\paragraph{GFT for directed graphs}
For directed graphs, the adjacency matrix is no longer symmetric -- $\mathbf{A}_{ij}$ is not necessarily equal to $\mathbf{A}_{j\,i}$ -- which does not automatically imply that the reference operator $\mathbf{R}$ is not symmetric (e.g., the case $\mathbf{R}=\mathbf{Q}$ below).  This case is of great interest in some applications where the graph is naturally directed such as hyperlink graphs (there is a directed edge between website $i$ and website $j$ if there is a hyperlink in website $i$ directing to $j$). In directed graphs, the degree of node $i$ is separated in its out-degree, $d_{\text{out}}=\sum_j A_{ij}$, and its in-degree, $d_{\text{in}}=\sum_j A_{ji}$. 
\\

\noindent\textbf{Some straightforward approaches [not symmetric].} It is possible to readily transpose the previous notions to the directed case, choosing either $\mathbf{D_{out}}$ or $\mathbf{D_{in}}$ to replace $\mathbf{D}$ in the different formulations: e.g. $\mathbf{L}=\mathbf{D_{{in}}} - \mathbf{A}^\top$ as in~\cite{Singh_SPCOM2016}, $\mathbf{L_{rw}}=\mathbf{I}- \mathbf{D_{out}}^{-1}\mathbf{A}$ as in~\cite{Sevi_SPIE17},  
$\mathbf{L_n}=\mathbf{I}-\mathbf{D_{{out}}}^{-\frac{1}{2}}\mathbf{A}\mathbf{D_{{out}}}^{-\frac{1}{2}}$.
A notable choice is to directly use $\mathbf{L_d} = \mathbf{I} - \frac{\mathbf{A}^\top}{||\mathbf{A}||_2}$ as in~\cite{sandryhaila_TSP2013} (we recall that $\mathbf{L_d}$ and $\mathbf{R}=\mathbf{A}^\top$ are equivalent for they share the same eigenvectors and thus, define the same GFT). 
These matrices are no longer strictly speaking Laplacians as they are no longer SDP, but one may nonetheless  consider them as reference operators defining possible GFTs. 
Note that these definitions entail to choose (rather arbitrarily) either $\mathbf{D_{{out}}}$ or $\mathbf{D_{{in}}}$ in their formulations, with the notable exception of $\mathbf{L_d}$. In fact, $\mathbf{L_d}$ naturally generalizes  to the directed case and is a classical choice of $\mathbf{R}$ in this context~\cite{sandryhaila_TSP2013}. 
\\

\noindent\textbf{Chung's directed Laplacian [symmetric].} A less common approach in the graph signal processing community is the one provided by the directed Laplacians introduced by Chung~\cite{chung_laplacians_2005}. To define these Laplacians appropriate to the directed case, let us first introduce the random walk transition matrix (or operator) defined as $\mathbf{P}=\mathbf{D_{out}^{-1}A}$. It admits a stationary probability\footnote{Assuming the random walk is ergodic, i.e. irreducible and non-periodic.} $\bm{\pi} \in \mathbb{R}_+^N$ such that $\bm{\pi}^\top \mathbf{P} = \bm{\pi}^\top$. Writing $\bm{\Pi}=\text{diag}(\bm{\pi})$, Chung defines the following two directed Laplacians:
\begin{align}
 \label{eq:DirectedLaplacian-combinatorial}
\mathbf{Q}&=\bm{\Pi}-\frac{\bm{\Pi}\mathbf{P}+\mathbf{P}^\top\bm{\Pi}}{2},\\
\label{eq:DirectedLaplacian-normalised}
\mathbf{Q_n}&=\bm{\Pi}^{-\frac{1}{2}}\mathbf{Q}\bm{\Pi}^{-\frac{1}{2}}=\mathbf{I}-\frac{\bm{\Pi}^{\frac{1}{2}}\mathbf{P}\bm{\Pi}^{-\frac{1}{2}}+\bm{\Pi}^{-\frac{1}{2}}\mathbf{P}^\top\bm{\Pi}^{\frac{1}{2}}}{2}.
\end{align}
Both the combinatorial $\mathbf{Q}$ and the normalized $\mathbf{Q_n}$ directed Laplacians verify the properties of Laplacian matrices:
SDP, negative (or null) entries everywhere except on the diagonal and real symmetric. 
It is easy to see that (\ref{eq:DirectedLaplacian-combinatorial}) and (\ref{eq:DirectedLaplacian-normalised}) generalise the definitions of the undirected case since for an undirected graph, 
$\bm{\Pi} = \mathbf{D}$, $\bm{\Pi} \mathbf{P} = \mathbf{A}$; hence $\mathbf{Q}$ is the combinatorial Laplacian $\mathbf{L}$, and $\mathbf{Q_n}$ is its normalized version $\mathbf{L_n}$. 
\\

\noindent\textbf{Other possible definitions of the reference operator.} The previous definitions of $\mathbf{L_{rw}}$ and $\mathbf{L_d}$ for undirected graphs may also be generalized to the directed Laplacian framework to obtain:
\begin{align}
 \mathbf{Q_{rw}}=\mathbf{I}-\frac{\mathbf{P}+\bm{\Pi}^{-1}\mathbf{P}^\top\bm{\Pi}}{2} ~~~\text{ and }~~~
\mathbf{Q_d}=\mathbf{I}-\frac{\bm{\Pi}\mathbf{P}+\mathbf{P}^\top\bm{\Pi}}{||\bm{\Pi}\mathbf{P}+\mathbf{P}^\top\bm{\Pi}||_2}.
\end{align}

\noindent\textbf{Additional notes.} Other GFTs for directed graphs were proposed via the Hermitian Laplacian as introduced in \cite{Yu_AMC_2015}, which generalizes $\mathbf{A}^\top$. A very different approach is to construct a Graph Fourier basis directly from an optimization scheme, requiring some notion of smoothness, or generalization of it -- see \cite{Sardellitti_JSTSP17,Shafipour_GlobalSIP2017}. We will not consider this recent approach here.

\subsection{Frequencies of graph signals}
\label{chapfilt:ssec13}

To complement the notion of GFT, one needs to introduce some frequency analysis of the Fourier modes on the graph.
The general way of doing so is to compute how fast a mode oscillates on the graph, and the tool of preference
is to compute their variations along the graph. 
Let us first note the following facts:
\begin{itemize}
 \item \textbf{In the undirected case}, $\mathbf{L}$ is semi-definite positive (SDP). In fact, one may write:
 \begin{align}
 \label{chapfilt:eq:L_variation}
  V_\mathbf{L}(\mathbf{x}) = \mathbf{x}^\top\mathbf{Lx} = \frac{1}{2} \sum_{(i,j)\,\in\, \mathcal{E}}  \mathbf{A}_{ij} \big( x_i - x_j \big)^2 \geq 0.
  \end{align}
 This function is also called the Dirichlet form.
  Similarly, $\mathbf{L_n}$ is also SDP:
  \begin{equation}
  \label{chapfilt:eq:Ln_variation}
 V_\mathbf{L_n}(\mathbf{x}) = \mathbf{x}^\top\mathbf{L_n x} = \frac{1}{2} \sum_{(i,j)\,\in\, \mathcal{E}}  \mathbf{A}_{ij} \left( \frac{x_i}{\sqrt{d_i}} - \frac{x_j}{\sqrt{d_j}} \right)^2 \geq 0. 
\end{equation}
As far as we know, the Dirichlet forms of $\mathbf{L_{rw}}$ and $\mathbf{L_d}$ do not have such a nice formulation as a sum of local quadratic variations over all edges of the graph. They are nevertheless SDP because: {\em i)}~$\mathbf{L_{rw}}$ and $\mathbf{L_n}$ have the same spectrum; {\em ii)}~the symmetry of $\mathbf{L_d}$ implies real eigenvalues, and the maximum eigenvalue of $\mathbf{A}/||\mathbf{A}||_2$ being $1$ by definition of the norm, the minimum eigenvalue of $\mathbf{L_d}$ is 0.
\item \textbf{In the directed case}, all directed Laplacians $\mathbf{Q}$, $\mathbf{Q_n}$, $\mathbf{Q_{rw}}$ and $\mathbf{Q_d}$ are SDP due to similar arguments. $\mathbf{Q}$ and $\mathbf{Q_n}$ also have Dirichlet forms in terms of a sum of local quadratic variations, e.g.:
\begin{align}
\label{chapfilt:eq:Q_variation}
 V_\mathbf{Q}(\mathbf{x}) = \mathbf{x}^\top\mathbf{Q}\mathbf{x}
 = \frac{1}{2}\sum_{(i,j)\,\in\, \mathcal{E}} \pi_i \mathbf{P}_{ij}  \big(x_i-x_j \big)^2 \geq 0.
\end{align}
The other reference operators $\mathbf{L}=\mathbf{D_{{in}}} - \mathbf{A}^\top$, $\mathbf{L_{rw}}=\mathbf{I}- \mathbf{D_{out}}^{-1}\mathbf{A}$,  
$\mathbf{L_n}=\mathbf{I}-\mathbf{D_{{out}}}^{-\frac{1}{2}}\mathbf{A}\mathbf{D_{{out}}}^{-\frac{1}{2}}$, $\mathbf{L_d} = \mathbf{I} - \frac{\mathbf{A}^\top}{||\mathbf{A}||_2}$  are not SDP as their eigenvalues may be complex. Nevertheless, the real part of their eigenvalues are always non-negative. This is quite clear for $\mathbf{L}_d$. For $\mathbf{L}=\mathbf{D_{{in}}} - \mathbf{A}^\top$, as the sum of row $i$ of $\mathbf{A}^\top$ is equal to $\mathbf{d_{in}}(i)$, Gershgorin circle theorem ensures that all eigenvalues of $\mathbf{L}$ are non-negative. For $\mathbf{L_{rw}}$, as  $\mathbf{P}=\mathbf{D_{out}}^{-1}\mathbf{A}$ is a stochastic matrix, the Perron-Frobenius theorem ensures that its  eigenvalues are in the disk of radius 1 in the complex plane, hence the real part of $\mathbf{L_{rw}}$'s eigenvalues are non-negative. As $\mathbf{L_{rw}}$ and $\mathbf{L_{n}}$ have the same set of eigenvalues, it is also true for $\mathbf{L_n}$. 
\end{itemize}
To sum up, all the reference operators considered are either SDP (with real non-negative eigenvalues) or have eigenvalues whose real component is non-negative.

\begin{definition}[Graph frequency]  \label{def:frequency}
Let $\mathbf{R}$ be a reference operator. If its eigenvalues are real, the generalized graph frequency $\nu$ of a graph Fourier mode $\mathbf{u}_k$ is:
\begin{equation}
\nu(\mathbf{u}_k) = \lambda_k\geq0.
\end{equation}
If its eigenvalues are complex, two different definitions of the generalized graph frequency $\nu$ of a graph Fourier mode $\mathbf{u}_k$ exist: 
\begin{equation}
\nu(\mathbf{u}_k) = \text{Re}(\lambda_k)\geq0 ~~~\text{ or }~~~\nu(\mathbf{u}_k) = |\lambda_k|\geq0.
\end{equation}
\end{definition}

\noindent\textbf{Remarks.} In the complex eigenvalues case, it is a matter of choice wether we consider the imaginary part of the eigenvalues or not. There is no current consensus on this question. Also, note that all eigenvectors associated to the same $\lambda_k$ have the same frequency. We thus define the frequency of an eigenspace, denoted by $\nu(\lambda_k)$.  \\

\noindent\textbf{Justification: the link between frequency and variation.} Two types of variation measures have been considered in the literature to show the consistency between this definition of graph frequencies and a notion of oscillation over the graph. The first one is based on the quadratic forms of the Laplacian operators. For instance, in the undirected case with the combinatorial Laplacian, Eq.~\eqref{chapfilt:eq:L_variation} applied to any normalized Fourier mode $\mathbf{u}_k$ defined from $\mathbf{L}$ reads:
\begin{align}
 V_\mathbf{L}(\mathbf{u}_k) = \mathbf{u}_k^\top\mathbf{Lu}_k = \frac{1}{2} \sum_{(i,j)\,\in\, \mathcal{E}}  \mathbf{A}_{ij}\, \big( u_k(i) - u_k(j) \big)^2=\lambda_k \,||\mathbf{u}_k||_2^2 = \lambda_k.
\end{align}
The larger the local quadratic variations of $\mathbf{u}_k$, the larger its frequency $\lambda_k$. Eqs.~\eqref{chapfilt:eq:Ln_variation} and \eqref{chapfilt:eq:Q_variation} (as well as its counterpart for $\mathbf{Q_n}$) enable to make this variation-frequency link for $\mathbf{L_n}, \mathbf{Q}$ and $\mathbf{Q_n}$. 
The second general type of variation that has been defined~\cite{sandryhaila_TSP14} is the total variation between a signal and its shifted version on the graph (where ``shifting'' a signal is understood as applying the adjacency matrix to it). For instance, in the case of $\mathbf{L_d}$, the associated variation reads\footnote{In~\cite{sandryhaila_TSP14}, the  $\ell_1$ norm is used, but the $\ell_2$ norm can be used equivalently: this is a matter of how one wants to normalize the eigenvectors. In this paper, we consider the classical Euclidean norm, hence $\ell_2$.}:
\begin{equation}
\label{eq:TV_A}
{V}_{\mathbf{L_d}}(\mathbf{x}) =   \left|\left|\mathbf{x} - \frac{1}{|\mu_{max}|} \mathbf{A}^\top \mathbf{x} \right|\right|_2 =\left|\left|  \mathbf{L_d x} \right|\right|_2 ,
\end{equation}
where $\mu_i$ designate the eigenvalues of $\mathbf{A}$ and $\mu_{max}$  the one of maximum magnitude. 
The  variation of the graph Fourier mode $\mathbf{u}_k$ from $\mathbf{L_d}$ thus reads: 
\begin{align}
 {V}_{\mathbf{L_d}}(\mathbf{u}_k)= \left|\left|  \mathbf{L_d u}_k \right|\right|_2 = |\lambda_k|\; ||\mathbf{u}_k||_2 = |\lambda_k|. 
\end{align}

The larger the total variation of $\mathbf{u}_k$, that is, the further is $\mathbf{u}_k$ from its shifted version along the graph, the larger its frequency\footnote{There is a direct correspondence between $\lambda_i$, the eigenvalues of $\mathbf{L_d}$, 
and $\mu_i$: $\lambda_i=1-\mu_i/|\mu_{max}|$. We thus recover the results in~\cite{sandryhaila_TSP14}: the closer is $\mu_k$ from $|\mu_{max}|$ in the complex plane, the smaller the  total variation of the associated Fourier mode $\mathbf{u}_k$.}. For $\mathbf{L}=\mathbf{D_{in}}-\mathbf{A}^\top$, a similar approach detailed in~\cite{Singh_SPCOM2016} links the variation of 
$\mathbf{u}_k$ to its frequency $|\lambda_k|$. 

It may also happen that for some operator $\mathbf{R}$, none of these two types of variations (quadratic forms or total variation) show natural. Then, one may use the variation $V_{\mathbf{R'}}$ based on another related operator $\mathbf{R'}$ to define frequencies. For instance, in~\cite{Sevi_SPIE17,Sevi_GRETSI17}, authors considered the random walk Laplacians $\mathbf{R}=\mathbf{P}=\mathbf{D_{{out}}}^{-1}\mathbf{A}$ as the reference operator to define the GFT, while the directed combinatorial Laplacian $\mathbf{R'}=\mathbf{Q}$ is used to measure the variations. With these choices, they showed that $V_{\mathbf{Q}}(\mathbf{u}_k)$ is equal to $\text{Re}(\lambda_k)$ up to a normalization constant. Another example of such a case is in \cite{BenAlaya:2016} where one of the reference operators to build filters is the isometric translation introduced in \cite{Girault2015} while the variational operators are built upon the combinatorial Laplacian. Note finally that other notions of variations like the Hub authority score can be drawn from the literature. We refer the reader to~\cite{Anis_TSP16} for details on that, as well as a complementary discussion on GFTs and their related variations. 
\\

Definition~\ref{def:frequency} associates graph frequencies only to graph Fourier modes. For an arbitrary signal, we have the following definition
of its frequency analysis:
\begin{definition}[Frequency analysis]  \label{def:frequencySpectrum}
The frequency analysis of any graph-signal $\mathbf{x}$ on $\mathcal{G}$ is given by its components
$(\mathbb{F}_\mathcal{G} \mathbf{x} )_k$ at frequency $\nu(\mathbf{u}_k)$, as given in Definition~\ref{def:frequency}. 
\end{definition}

\begin{figure}\centering
\begin{minipage}{0.67\textwidth}
 a) \includegraphics[width=0.42\textwidth]{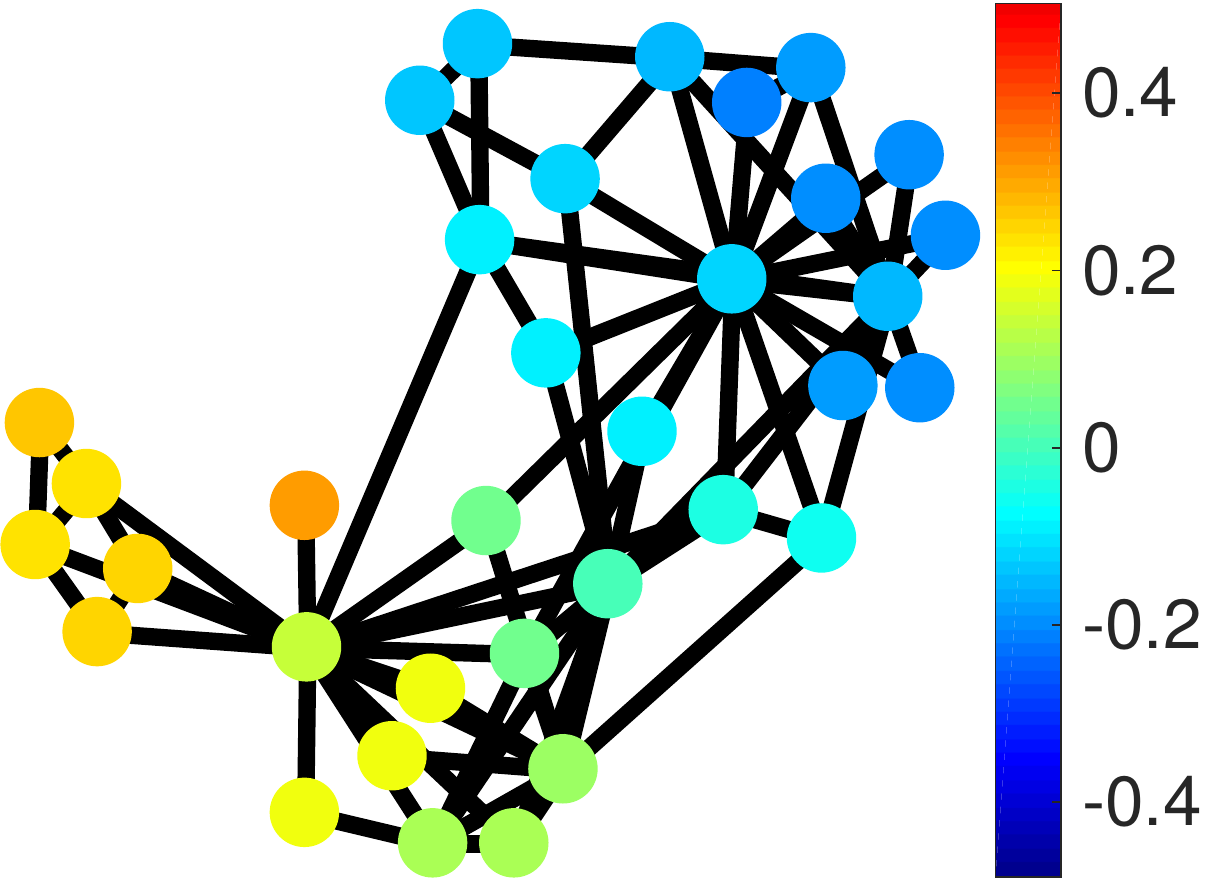}\hfill
 b)  \includegraphics[width=0.42\textwidth]{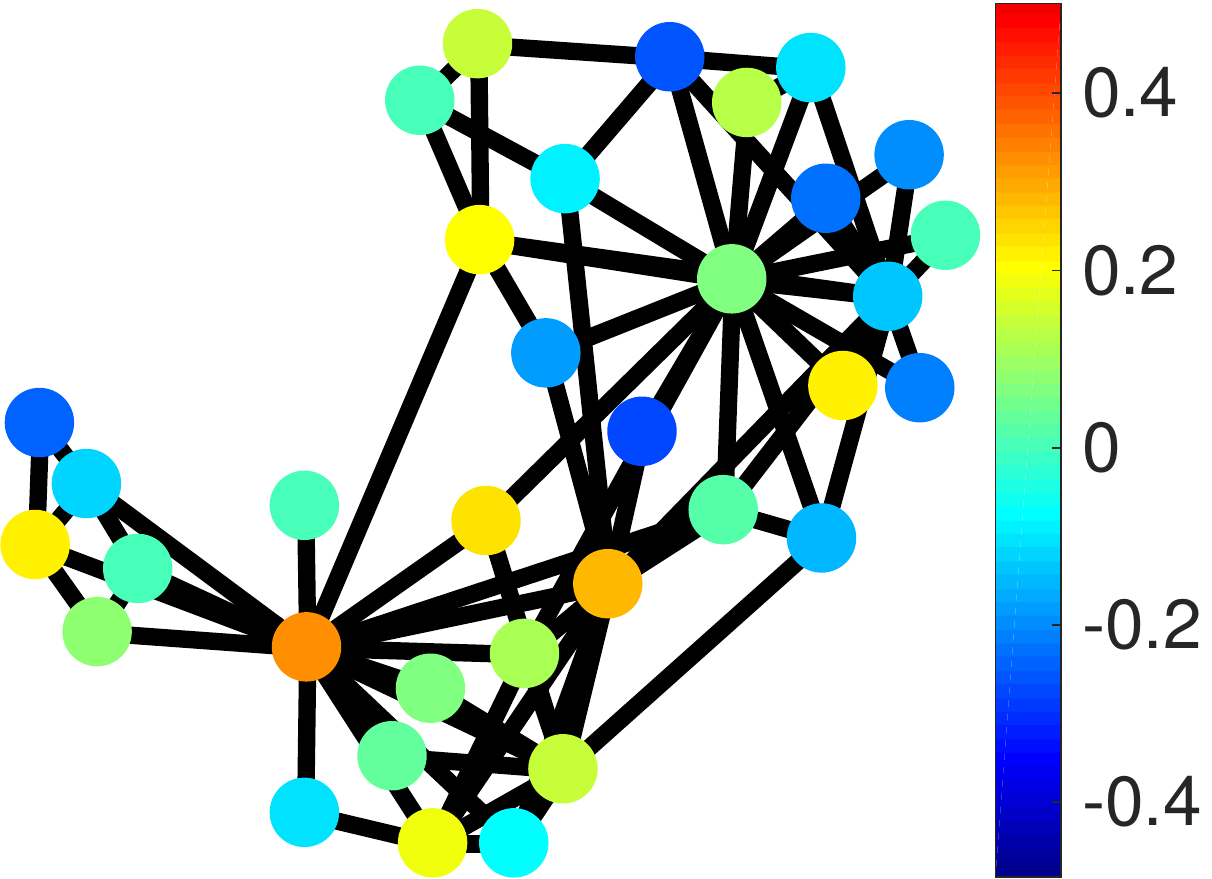}
\end{minipage}\hfill
\begin{minipage}{0.28\textwidth}
  c)\includegraphics[width=0.9\textwidth]{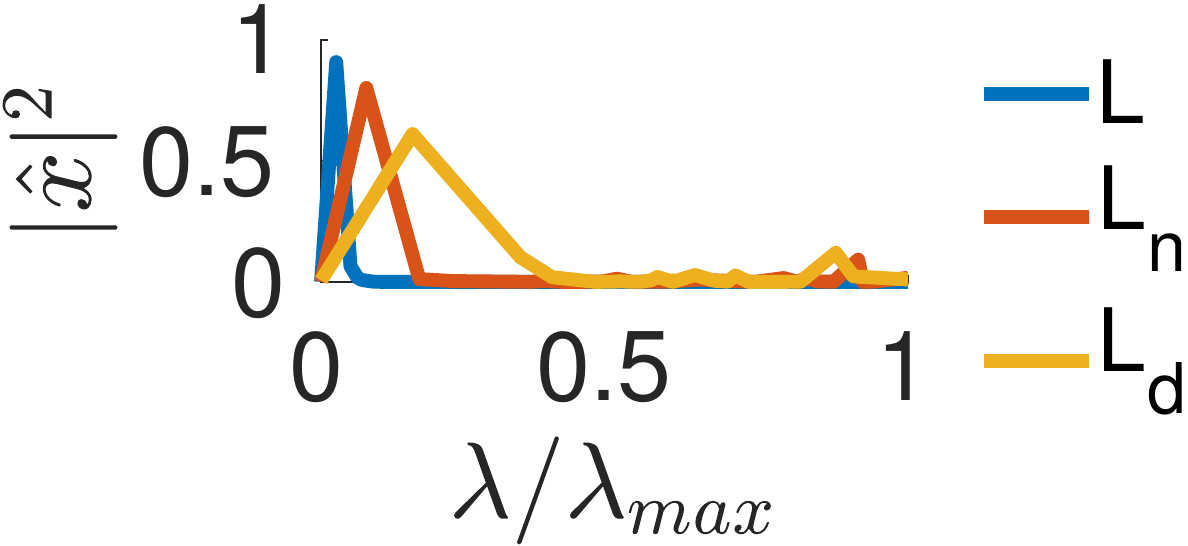}\\
 d)\includegraphics[width=0.9\textwidth]{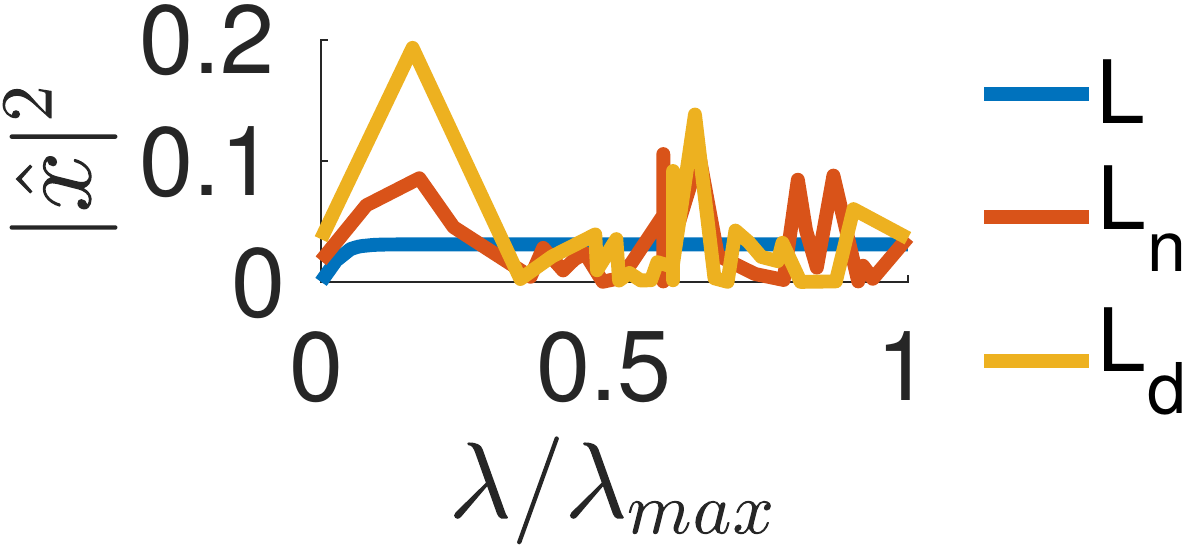}
\end{minipage}
 \caption{\label{chapfilt:fig:illus} \textbf{Two graph signals and their GFTs.} Plots a) and b) represent respectively, a low-frequency and a high-frequency graph signal on the binary Karate club graph~\cite{karate_club}. Plots c) and d) are their corresponding GFTs computed for three reference operators: $\mathbf{L}$, $\mathbf{L_n}$ and $\mathbf{L_d}$ (equivalent to the GFT defined via the adjacency matrix). }
\end{figure}

\subsection{Implementation and illustration}

\noindent\textbf{Implementation.} The implementation of the GFT requires to diagonalize $\mathbf{R}$,  costing $\mathcal{O}(N^3)$ operations in general, and $\mathcal{O}(N^2)$ memory space to store $\mathbf{U}$. Then, applying $\mathbf{U}$ to a signal $\mathbf{x}$ to obtain $\hat{\mathbf{x}}=\mathbf{U}^{-1}\mathbf{x}$ costs $\mathcal{O}(N^2)$ operations. These costs are prohibitive for large graphs ($N\gtrsim 10^4$ nodes). Recent works investigate how to reduce these costs, by tolerating an approximation on $\hat{\mathbf{x}}$. In the cases where $\mathbf{R}$ is symmetric, authors in~\cite{magoarou_approximate_2017, ASILOMAR2017_FGFT} suggest to approximate $\mathbf{U}$ by a product of $J=\mathcal{O}(N\log{N})$ Givens rotations, using a truncated Jacobi algorithm. The resulting approximated fast GFT requires $\mathcal{O}(N^2\log^2{N})$ operations to compute the Givens rotations, and $\mathcal{O}(N\log{N})$ operations to compute the approximate GFT of $\mathbf{x}$. The difficulty in designing fast GFTs boils down to the difficulty of deciphering eigenvalues that are very close to one another. This difficulty disappears once we consider smooth filtering operations that are much easier to efficiently approximate, as we will see in the next section.
\\

\noindent\textbf{Illustrations.} To illustrate the GFT and the notion of frequency, we show in Fig.~\ref{chapfilt:fig:illus} two graph signals on the Karate club graph~\cite{karate_club}, corresponding to instances of a low frequency and a high frequency signal, respectively. We also show their GFTs, computed for three different choices of $\mathbf{R}$: $\mathbf{L}$, $\mathbf{L_n}$ and $\mathbf{L_d}$. 

The choice of normalization and the choice to take explicitly the degree matrix into account or not in the definition of $\mathbf{R}$ has a quantitative impact on the GFTs. Nevertheless, qualitatively, a graph signal that varies slowly (resp. rapidly) along any path of the graph is low-frequency (resp. high-frequency). In Fig.~\ref{chapfilt:fig:illus_bis}, we show how the GFT is not only sensitive to the graph signal but also to the underlying graph structure. 

\begin{figure}\centering
\begin{minipage}{\textwidth}\hfill
 a) \includegraphics[width=0.30\textwidth]{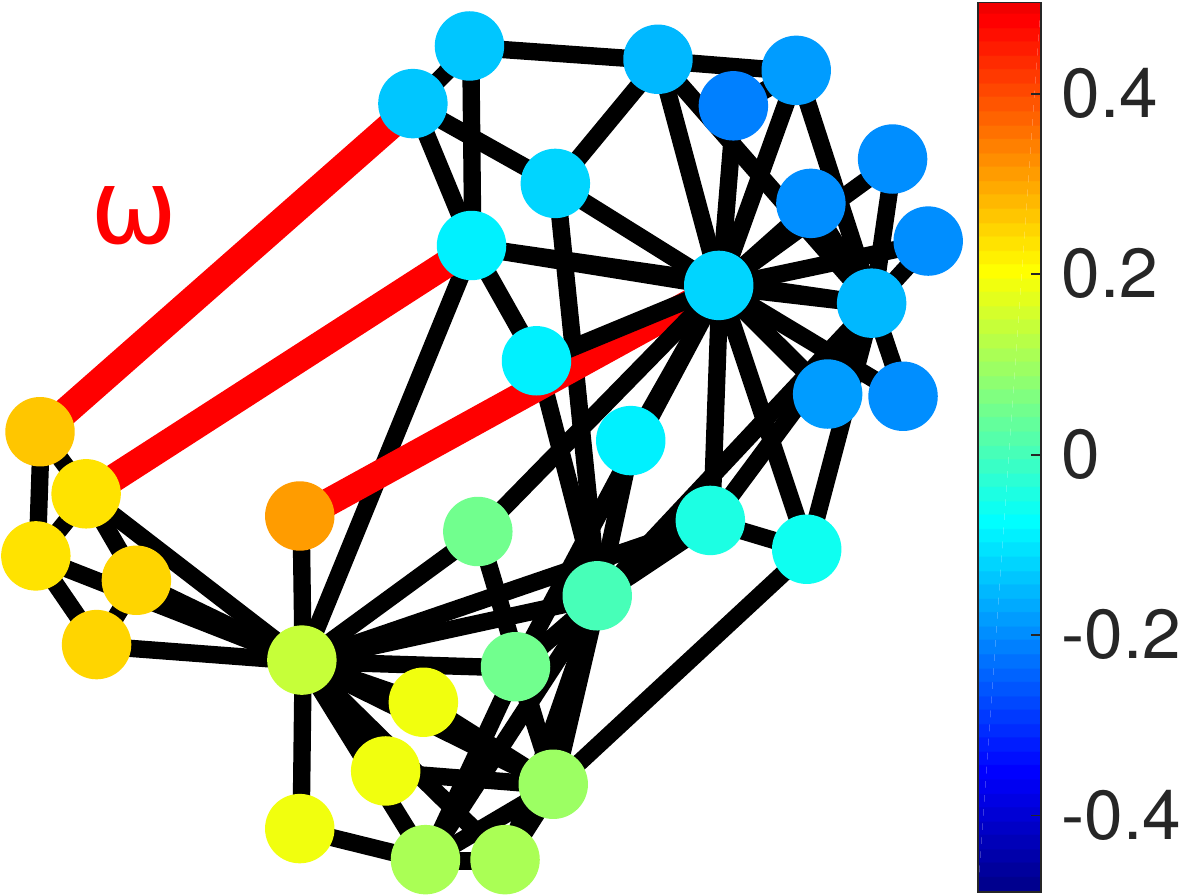}\hfill
 b)  \includegraphics[width=0.42\textwidth]{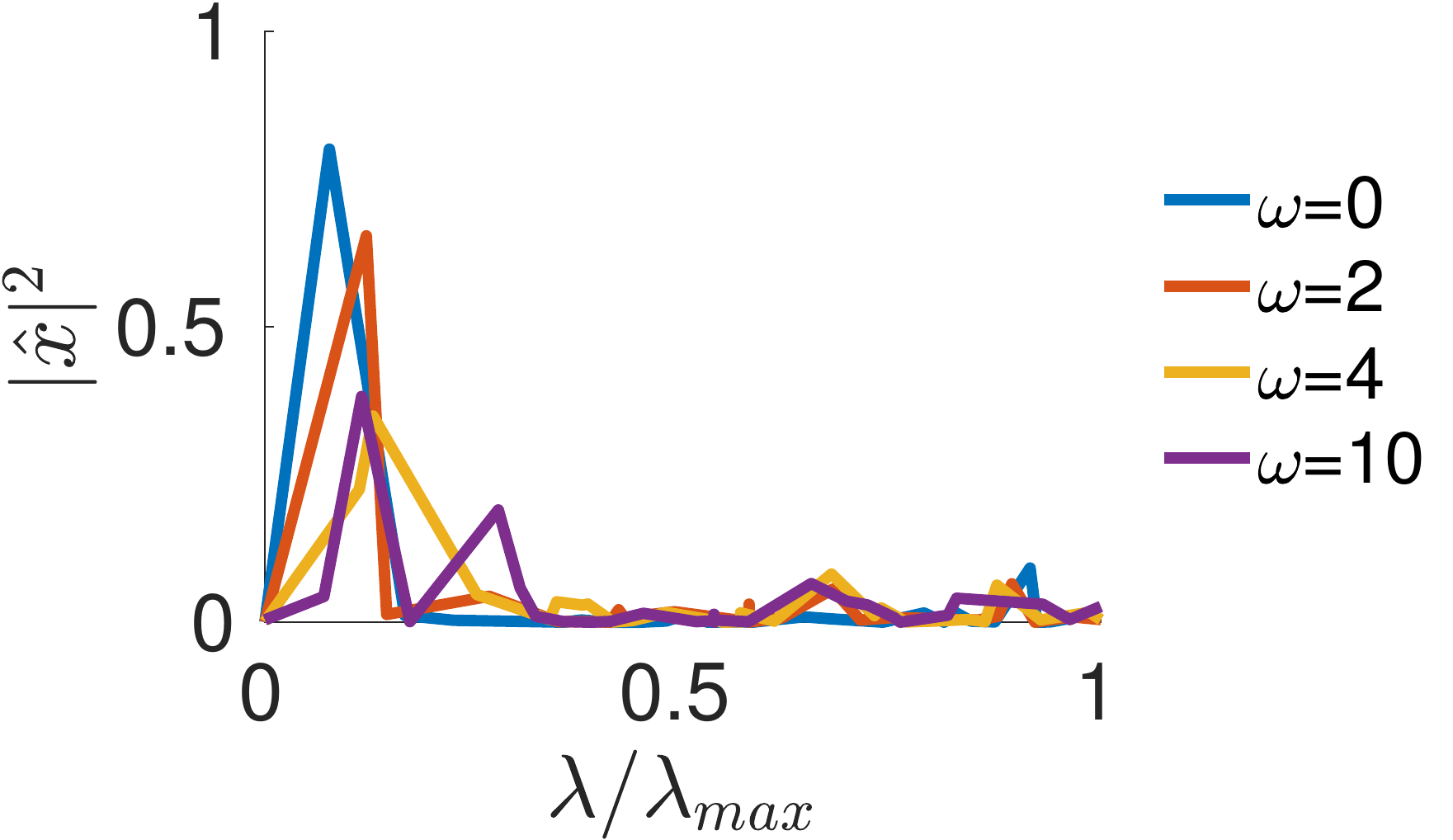}\hfill~
 \end{minipage}
 \caption{\label{chapfilt:fig:illus_bis}\textbf{Sensitivity of the GFT to the graph's topology.} Plot a) represents the same low-frequency graph signal as in Fig.~\ref{chapfilt:fig:illus}-a), but the underlying graph structure is altered by adding three edges with same weight $\omega$ (in red). Plot b) represents the variation of its GFT (here choosing $\mathbf{R}=\mathbf{L_n}$) with respect to the edges' weight $\omega$. }
\end{figure}

\section{Graph filters}
\label{chapfilt:sec2}
In this section, we assume that the reference operator $\mathbf{R}$ of the graph on which we wish to design filters is  diagonalizable in $\mathbb{C}$ as in the previous section. The order of the eigenvalues and eigenvectors is chosen frequency-increasing. That is, given a choice of frequency definition (either $\nu(\lambda_k)=\text{Re}(\lambda_k)\in\mathbb{R}^+$ or $\nu(\lambda_k)=|\lambda_k|\in\mathbb{R}^+$), one has $\nu(\mathbf{u}_1)\leq\nu(\mathbf{u}_2)\leq\ldots\leq\nu(\mathbf{u}_N)$. 


\subsection{Definition of graph filters}
\label{chapfilt:ssec21}

The reference operator $\mathbf{R}$ has an eigendecomposition as in Eq.~(\ref{eq:eigendecomp}), and it can also be written as a sum of projectors on all its eigenspaces:
\begin{align}
\label{chapfilt:eq:R_as_projectors}
 \mathbf{R} = \sum_{\lambda} \lambda~\text{\textbf{Pr}}_{\lambda},
\end{align}
where the sum is on all different eigenvalues $\lambda$ and $\text{\textbf{Pr}}_{\lambda}$ is the projector on the eigenspace associated to eigenvalue $\lambda$, i.e.: 
$$\text{\textbf{Pr}}_{\lambda}=\sum_{\lambda_{k}=\lambda} \mathbf{u}_{k}\mathbf{v}_{k}^\top.$$

\begin{definition}[wide-sense definition of a graph filter]
\label{def:wide-sense}
 The most general definition of a graph filter is an operator that acts separately on all the eigenspaces of $\mathbf{R}$, depending on their eigenvalue $\lambda$. Mathematically, any function 
\begin{align}
\label{chapfilt:eq:h}
 h:
 \mathbb{C} &\rightarrow \mathbb{R}\\
 \lambda&\rightarrow h(\lambda).
\end{align}
defines a graph filter $\mathbf{H}$ such that
\begin{align}
\label{eq:wide-sense-filter}
 \mathbf{H} = \sum_{\lambda} h(\lambda)~\text{\textbf{Pr}}_{\lambda}= \sum_{k} h(\lambda_k) \mathbf{u}_{k}\mathbf{v}_{k}^\top.
\end{align}
\end{definition}

The decomposition of the filter on each eigenspace of $\mathbf{R}$ with eigenvalue $\lambda$ is associated to a filtering weight $h(\lambda)$ that attenuates or increases the importance of this eigenspace in the decomposition of the signal of interest. In fact, one may write the action of $\mathbf{H}$ on a graph signal $\mathbf{x}$ as:
 \begin{align}
  \mathbf{Hx} = \sum_{\lambda} h(\lambda)~\text{\textbf{Pr}}_{\lambda} \;\mathbf{x}.
 \end{align}
%
For any function $g$, let us write $g(\mathbf{\Lambda})$ as a shorthand notation for $\text{diag}(g(\lambda_1)\ldots,g(\lambda_N))$. A graph filter can be written as:
\begin{align}
 \mathbf{H} = \mathbf{U}\,h(\mathbf{\Lambda})\,\mathbf{U}^{-1}.
\end{align}
Using functional calculus of operators, this is equivalently written as $\mathbf{H} = h(\mathbf{R})$, which calls for some interpretation remarks. In fact, this expression opens the way to interpret what is the action of a graph filter in the vertex domain. Firstly, note that applying $\mathbf{R}$ to a graph signal is in fact a \emph{local} computation on the graph: on each node, the resulting transformed signal is a weighted sum of the values of the original signal on its (direct) neighbours. 
Therefore, in the node space, a filter $\mathbf{H}=h(\mathbf{R})$ can be interpreted as an operator that weights the information on the signal transmitted through edges of the graph, the same way classical filters are built on the basic operation of time-shift. This fundamental analogy shift/reference operator was first made in~\cite{sandryhaila_TSP2013} with the adjacency matrix and was particularly pertinent in that case. Yet, in the GSP literature, some authors extended this analogy to all reference operators, using sometimes the generic term ``graph shift operators" for their designation. 



Now, to illustrate the notion of graph filtering in the spectral domain, consider the graph signal $\mathbf{x}=\sum_k \alpha_k \mathbf{u}_k$
and $\mathbf{y} = \mathbf{Hx} $. The $k$-th Fourier component of $\mathbf{x}$ being, by construction, $\hat{\mathbf{x}}_k = \mathbf{v}_k^\top \mathbf{x} = \alpha_k$, its filtered version reads:
\begin{align}
\mathbf{y} =  \mathbf{Hx} =  \sum_{k}  h(\lambda_k)\alpha_k \mathbf{u}_k.
\end{align}
Hence the $k$-th Fourier component of the filtered signal is $\hat{\mathbf{y}}_k  = h(\lambda_k)\,\alpha_k$, recalling the classical interpretation of filtering as a multiplication in the Fourier domain. 
We call $h(\lambda)$ the frequency response of the filter.

In many cases, it is more convenient and natural to restrict ourselves to functions $h$ that associate the same real number to all values of $\lambda\in\{\lambda \text{ s.t. } \nu(\lambda)=\nu\in\mathbb{R}^+\}$. That is, considering any two eigenvalues $\lambda_1$ and $\lambda_2$ associated to the same frequency $\nu(\lambda_1)=\nu(\lambda_2)$, we restrict ourselves to functions $h$ such that $h(\lambda_1)=h(\lambda_2)$. This entails the following narrowed-sense definition of a graph filter.

\begin{definition}[narrowed-sense definition of a graph filter]
\label{def:narrow-sense}
 Any function 
\begin{align}
\label{chapfilt:eq:h}
 h:
 \mathbb{R}^+ &\rightarrow \mathbb{R}\\
 \nu&\rightarrow h(\nu).
\end{align}
defines a narrowed-sense graph filter $\mathbf{H}$ such that
\begin{align}
 \mathbf{H} = \sum_{\lambda} h(\nu(\lambda))~\text{\textbf{Pr}}_{\lambda} = \mathbf{U}h(v(\mathbf{\Lambda}))\mathbf{U}^{-1}.
\end{align}
\end{definition}

\noindent\textbf{Remark.} If the eigenvalues are real, both definitions \ref{def:wide-sense} and \ref{def:narrow-sense} are equivalent since $\nu(\lambda)=\lambda\in\mathbb{R}^+$. 
\\



\noindent\textbf{Examples of narrowed-sense filters:}
\begin{itemize}
 \item The constant filter equal to $c$: $h(\nu) = c$. In this case, $h(\nu(\mathbf{\Lambda})) = c\mathbf{I}$ and $\mathbf{H} = c\mathbf{I}$: all frequencies are allowed to pass, and no component is filtered out. 
 \item The Kronecker delta in $\nu^*$: $h(\nu) = \delta_{\nu, \nu^*}$. If there exists one (or several) eigenvalues  $\lambda$ of $\mathbf{R}$ such that $\nu(\lambda)=\nu^*$, then $\mathbf{H} = \sum_{\lambda \text{ s.t. }\nu(\lambda)=\nu^*}\text{\textbf{Pr}}_{\lambda}$. If not, then $\mathbf{H} = \mathbf{0}$. For this filter, only the frequency $\nu^*$ is allowed to pass. 
 \item The ideal low-pass with cut-off frequency $\nu_c$: $h(\nu) = 1$ if $\nu\leq \nu_c$, and $0$ otherwise. In this case:
 $\mathbf{H} = \sum_{\lambda, \text{s.t.} \nu(\lambda)\leq \nu_c} \text{\textbf{Pr}}_\lambda$, \textit{i.e.}, only frequencies up to $\nu_c$ are allowed to pass.
 \item The heat kernel $\displaystyle h(\nu)=\exp^{-\nu/\nu_0}$: the weight associated to $\nu$ is exponentially  decreasing with the frequency $\nu$. Actually $\mathbf{y} =\mathbf{H}\,\mathbf{x}_0$ is the solution of the graph diffusion (or heat) equation (see \cite{shuman_SPMAG2013}) at time $t=1/\nu_0$ with initial condition $\mathbf{x}_0$. 
\end{itemize}

\subsection{Properties of graph filters}
\label{sec:prop_graph}
From now on, in order to simplify notations and concepts in this introductory chapter on graph filtering, we will restrict ourselves to symmetric reference operators, such as $\mathbf{L}$, $\mathbf{L_n}$, or $\mathbf{L_d}$ in the undirected case, or the directed Laplacians $\mathbf{Q}$ or $\mathbf{Q_n}$ in the directed case. In this case, the eigenvalues are real such that the frequency definition is straightforward $\nu(\lambda)=\lambda$, both filter definitions are equivalent such that the frequency response reads
\begin{align}
\label{chapfilt:eq:h}
 h:
 \mathbb{R}^+ &\rightarrow \mathbb{R}\\
 \lambda&\rightarrow h(\lambda),
\end{align}
and one may find a real orthonormal graph Fourier basis $\mathbf{U}$ such that $\mathbf{U}^{-1}=\mathbf{U}^\top$. A filtering operator associated to $h$ thereby reads:
\begin{align}
 \mathbf{H} = \mathbf{U} h(\mathbf{\Lambda})\mathbf{U}^\top\in\mathbb{R}^{N\times N}.
\end{align}
All results presented in the following may be (carefully) generalized to the unsymmetric case. 

\begin{definition}
 We write $\mathcal{C}_p$ the set of finite-order polynomials in $\mathbf{R}$:
 \begin{align}
\label{chapfilt:eq:filter_to_poly}
 \mathcal{C}_p = \left\{\mathbf{H} \text{ s.t. } \mathbf{H} = \sum_{i=0}^{n} a_i \mathbf{R}^i, \: \{a_i\}_{i=0,\ldots,n}\in\mathbb{R}^{n+1}, n\in\mathbb{N}\backslash\{+\infty\}\right\}.
\end{align}
\end{definition}

\begin{proposition}
$\mathcal{C}_p$ is equal to the set of graph filters. 
\end{proposition}
\begin{proof}
Consider $\mathbf{H}\in\mathcal{C}_p$. Then, defining $h(\lambda)= \sum_{i=0}^{n} a_i \lambda^i$, one has: $\mathbf{H} = \sum_{i=0}^{n} a_i \mathbf{R}^i = \mathbf{U}\sum_{i=0}^{n} a_i \mathbf{\Lambda}^i\mathbf{U}^\top = \mathbf{U}h(\mathbf{\Lambda})\mathbf{U}^\top$, \textit{i.e.}, $\mathbf{H}$ is a filter. Now, consider $\mathbf{H}$ a filter, \textit{i.e.}, there exists $h$ such that $\mathbf{H}=\mathbf{U}h(\mathbf{\Lambda})\mathbf{U}^\top$. Consider the polynomial $\sum_{i=0}^{N-1} a_i \lambda^i$ that interpolates through all pairs $(\lambda_i, h(\lambda_i))$. The maximum degree of such a polynomial is $N-1$ as there are maximum $N$ points to interpolate, and may be smaller if eigenvalues have multiplicity larger than one. Thereby, one may write:
$\mathbf{H}=\mathbf{U}h(\mathbf{\Lambda})\mathbf{U}^\top=\mathbf{U}\sum_{i=0}^{N-1} a_i \mathbf{\Lambda}^i\mathbf{U}^\top=\sum_{i=0}^{N-1} a_i \mathbf{R}^i$. Writing $n=N-1$, this means that $\mathbf{H}\in\mathcal{C}_p$. 
\end{proof}

\noindent\textbf{Consequence:} An equivalent definition of a graph filter is a polynomial in $\mathbf{R}$ of maximal degree $N-1$. 

\begin{definition}
 We write $\mathcal{C}_d$ the set of all diagonal operators in the graph Fourier space:
\begin{align}
 \mathcal{C}_d = \{\mathbf{H}, \text{ s.t. } \mathbf{U}^{\top}\mathbf{HU} \text{ is diagonal}\}.
\end{align}
\end{definition}

\begin{proposition} 
\label{chapfilt:prop:C_d}The set of graph filters is included in $\mathcal{C}_d$. Both sets are equal iff  all eigenspaces of $\mathbf{R}$ are of dimension one (\textit{i.e.}, all eigenvalues are of multiplicity one).
\end{proposition}
\begin{proof}
By definition of graph filters, they are included in $\mathcal{C}_d$. 
Now, in general, an element of $\mathcal{C}_d$ is not necessarily a graph filter. In fact, given  $\mathbf{H}\in\mathcal{C}_d$, all diagonal entries of $\mathbf{U}^\top\mathbf{HU}$ may be chosen independently, which is not the case for the diagonal entries of $h(\mathbf{\Lambda})$ corresponding to the same eigenspace. Thus, both sets are equal iff all eigenspaces are of dimension one. 
 \end{proof}
 
\noindent\textbf{Consequence:} In the case where all eigenvalues are simple, an equivalent definition of a graph filter is a diagonal matrix in the graph Fourier basis.

\begin{definition}
We write $\mathcal{C}_c$ the set of matrices that commute with $\mathbf{R}$:
\begin{align}
 \mathcal{C}_c = \{\mathbf{H} \text{ s.t. } \mathbf{RH} = \mathbf{HR}\}.
\end{align}
\end{definition}
\begin{proposition}
 $\mathcal{C}_p\subseteq \mathcal{C}_c$. The equality holds iff all eigenvalues of $\mathbf{R}$ are simple. 
\end{proposition}
\begin{proof}
A polynomial in $\mathbf{R}$ necessarily commutes with $\mathbf{R}$, thus  $\mathcal{C}_p\subseteq \mathcal{C}_c$. Now, in general, an element in $\mathcal{C}_c$ is not necessarily in $\mathcal{C}_p$. In fact, if  $\mathbf{H}\in\mathcal{C}_c$, then $\mathbf{\Lambda}\mathbf{Y} = \mathbf{Y}\mathbf{\Lambda}$ with $\mathbf{Y}=\mathbf{U}^\top\mathbf{HU}$. Suppose $\mathbf{\Lambda}=\lambda\mathbf{I}$. In this case, commutativity does not constrain $\mathbf{Y}$ at all, thereby $\mathbf{H}$ is not necessarily in $\mathcal{C}_p$. Now, if we suppose that all eigenvalues have multiplicity one, \textit{i.e.}, all diagonal entries of $\mathbf{\Lambda}$ are different, the only solution for $\mathbf{Y}$ is to be a diagonal matrix, \textit{i.e.}, $\mathcal{C}_c=\mathcal{C}_d$. We showed in Prop.~\ref{chapfilt:prop:C_d} that the set of graph filters is equal to $\mathcal{C}_d$ iff all eigenvalues have multiplicity one, therefore:  $\mathcal{C}_c=\mathcal{C}_p$ iff all eigenvalues have multiplicity one.
\end{proof}

\noindent\textbf{Consequence:} In the case where all eigenvalues are simple, an equivalent definition of a graph filter is a linear operator that commutes with $\mathbf{R}$. 
\\



\subsection{Some designs of graph filters}

%
%

From the previous sections, it should now be clear that the frequency response of a filter $\mathbf{H}$ only alters the frequencies $\nu(\lambda)$ corresponding to the discrete set of eigenvalues of $\mathbf{R}$ (see Eq. (\ref{eq:wide-sense-filter})). More generally though, the frequency response of a graph filter can be defined over a continuous range of $\lambda$'s, leading to the notion of \textit{universal filter} design, i.e. a filter whose frequency response $h(\lambda)$ is designed for all $\lambda$'s  and not only adapted to the specific eigenvalues of $\mathbf{R}$. On the contrary, a \textit{graph-dependent filter} design depends specifically on these eigenvalues.
\medskip


\noindent \textbf{FIR filters.}
From the results of Section \ref{sec:prop_graph}, a natural class of graph filters is given in the form of Finite Impulse Response (FIR) filters, 
as by eq.~(\ref{chapfilt:eq:filter_to_poly}) with a polynomial of finite order, which realizes a weighted Moving Average (MA) filtering of a signal. Also, one can design any universal filter by fitting the desired response $h(\lambda)$ with a polynomial $\sum_{i=0}^n a_i \lambda^i$. The larger $n$ is, the closer the filter can approximate the desired shape. If the approximation is done only using the $h(\lambda_k)$, the design is graph-dependent, else it is universal if fitting some function function  $h(\lambda)$.

Let us then go back to the interpretation of $\mathbf{R}$ as a graph shift operator  in the node space and see how it operates for FIR filters. 
Applied to a graph signal, the terms $\mathbf{R}^i$ in eq.~(\ref{chapfilt:eq:filter_to_poly}), act as a $i$-hops local computation on the graph: on each node, the resulting filtered signal is a weighted sum of the values of the original signal lying in its $i$-th neighbourhood, that is, nodes attainable with a path of length $i$ along the graph. Then, like for classical signals, FIR filters only imply a finite neighbourhood of each nodes, and this will translate in Section \ref{ssec:chapfilt:ImplementationGFilt} into distributed, fast implementations of these filters. 
Still, FIR filters are usually poor at approximating filters with sharp changes of desired frequency response, as illustrated for instance in Fig. 1 of~\cite{isufi_autoregressive_2017}. 
\medskip

\noindent \textbf{ARMA filters.}
A more accurate and approximation of $h(\lambda)$ can be obtained with a rational design~\cite{isufi_autoregressive_2017,loukas_ARMA_2015,shi_infinite_2015}:
\begin{equation}
\label{chapfilt:eq:ARMAfilter}
h(\lambda) = \frac{ \sum_{i=0}^q b_i \lambda^i }{ 1 +  \sum_{i=1}^p a_i \lambda^i} = \frac{p_q(\lambda)}{p_p(\lambda)}.
\end{equation} 
Such a rational filter is called an Auto-Regressive Moving Average filter of order $(p,q)$ and is commonly noted ARMA$(p,q)$. 
Again, it is known from classical DSP that an ARMA design, being a IIR (Infinite Impulse Response) filter, is more versatile at approximating various
shapes of filters, especially with sharp changes in the frequency response. 
The filtering relation $\mathbf{y} = \mathbf{H} \mathbf{x}$ for ARMA filters can be written in the node domain as:
$  (1 +  \sum_{i=1}^p a_i \mathbf{R}^i) \mathbf{y} = (\sum_{i=0}^q b_i \mathbf{R}^i) \mathbf{x}$.
This ARMA filter expression will lead to the distributed implementation, discussed later in \ref{ssec:chapfilt:ImplementationGFilt}. For instance, for an ARMA(1,0) (i.e., an AR(1))  one will have to use:
$\mathbf{y}  = -a_1 \mathbf{R} \mathbf{y} + b_0 \mathbf{x}$.
\medskip

\noindent \textbf{Example:} A first design of low-pass graph filtering is given by the simplest least-square denoising problem, where the the Dirichlet form $\mathbf{x}^T \mathbf{R} \mathbf{x}$ is used as Tikhonov regularization promoting smoothness on the graph. Using the (undirected) Laplacian and assuming one observes $\mathbf{y}$, the filter is given by;
\begin{equation}
\label{eq:chapfilt:tikhonov}
\mathbf{x_*} = \mbox{arg}\min_{\mathbf{x}} || \mathbf{x}-\mathbf{y} ||_2^2 + \gamma \mathbf{x}^{\top} \mathbf{L} \mathbf{x}.
\end{equation} 
The solution is then given in the spectral domain (for $\mathbf{L}$) by $(\mathbb{F} \mathbf{x_*})_k = h_{AR(1)}(\lambda_k) (\mathbb{F} \mathbf{y})_k$
with $h_{AR(1)}(\lambda) = {1}/{(1+\gamma \lambda)}$. It turns out to be a (universal) AR(1) filter.
\medskip

\noindent \textbf{Design of coefficients.}
To design the coefficients of ARMA filters, the classical approach is to find the set of coefficients $a_i$ and $b_i$ to approximate the desired $h(\lambda)$ as a rational function. 
However, as recalled in \cite{isufi_autoregressive_2017,loukas_ARMA_2015}, the usual design in DSP are not easily transposed
to the GSP framework because the frequency response is given in terms of the $\lambda$'s, and not in terms 
of $j\omega$ or $e^{j\omega}$. 
 
Henceforth, it has been studied in \cite{isufi_autoregressive_2017,loukas_ARMA_2015} how to approximate the filter coefficients in a universal manner (i.e., with no specific reference to the graph spectrum) using a Shank's method: 
1) Determine the $a_i$ by first finding a polynomial approximation $P_h(\lambda)$ of $h(\lambda)$,
and solve the system of equations $p_p(\lambda) P_h(\lambda) = p_q(\lambda)$ to identify the $a_i$'s.
Then 2) solve the least-square problem to minimize $\int_\lambda | p_q(\lambda)/ p_p(\lambda) - h(\lambda) |^2 d\lambda$ w.r.t. $\lambda$ to find the $b_i$'s.

A second method is to approximate the filter response in a graph-dependent design, on the specific frequencies $\lambda_k$ only. To do so, the method in \cite{Liu_Isufi_Leus_2016}, instead of fitting the polynomial ratio,  solves the following optimization problem: 
\begin{equation}
\min_{\mathbf{a},\mathbf{b}} \sum_{k=0}^{N-1} \left | h(\lambda_k)  \left(1 +  \sum_{i=1}^p a_i \lambda_k^i\right) - \sum_{i=0}^q b_i \lambda_k^i \right|^2.
\end{equation} 
Using again a polynomial approximation $P_h(\lambda)$, the solution derives from the least square solution (see details in \cite{Liu_Isufi_Leus_2016}).
\medskip
 
\noindent \textbf{AR filters to model random processes.}
We do not discuss much in this chapter random processes on graphs ; for that see the framework to study stationary random processes on graphs in \cite{Perraudin_TSP2017,Marques_TSP2017}. 
Still, a remark can be done for the parametric modeling of random processes. 
As introduced in \cite{sandryhaila_TSP2013}, one can model a process on a graph
as the output of a graph filter, generally taken as an ARMA filter. 
Here, we discuss the case of AR filters. The linear prediction is written as:
\begin{equation}
\tilde{\mathbf{x}}  = \sum_{i=1}^p a_i \mathbf{R}^i \mathbf{x}.
\end{equation} 
The coefficients of this filter can directly be obtained by numerical inversion of
$\mathbf{x} = \mathbf{B} \mathbf{a}$ where $\mathbf{B}=(\mathbf{R} \mathbf{x}, \mathbf{R}^2 \mathbf{x},\ldots,\mathbf{R}^p \mathbf{x})$.
The solution is given by the pseudo-inverse: $\mathbf{a}^\# = (\mathbf{B}^{\top} \mathbf{B})^{-1} \mathbf{B}^{\top} \mathbf{x}$
which minimizes the squared error $||\mathbf{x}-\mathbf{B}\mathbf{a}||^2_2$, w.r.t to $\mathbf{a}$.
Another possibility would be to estimate the coefficients of the AR model by means of the orthogonality principle leading to the Yule-Walker (YW) equations, as follows:
\begin{equation}
	\mathbb{E}\left\{(\mathbf{R}^k \mathbf{x})^\top \mathbf{R}\mathbf{x})\right\}= \sum_{i=1}^p a_i   \mathbb{E}\left\{ (\mathbf{R}^i \mathbf{x})^\top \mathbf{R}^k \mathbf{x} \right\}=0,
\end{equation}
Depending to the structure of the reference operator $\mathbf{R}$, we have:
        \begin{enumerate}
    \item If $\mathbf{R}$ is symmetric ($\mathbf{R}^\top=\mathbf{R}$, as often for undirected graph), the autocorrelation function involved in this YW system is: $\rho_{\mathbf{x}}(m,n)= \mathbb{E}\left\{ (\mathbf{R}^m \mathbf{x})^\top \mathbf{R}^n \mathbf{x} \right \}=\gamma_{\mathbf{x}}(m+n)$;
    \item When $\mathbf{R}$ is unitary (for instance with $\mathbf{R}$ as the isometric operator from \cite{Girault2015}), the corresponding autocorrelation will be 
$\rho_{\mathbf{x}}(m,n)= \mathbb{E}\left\{ (\mathbf{R}^m \mathbf{x})^\top \mathbf{R}^n \mathbf{x} \right \}=\gamma_{\mathbf{x}}(n-m)$ and the usual techniques to solve the YW system can be used. 
\end{enumerate}
Experimentally, it was found that the isometric operator of \cite{Girault2015} or a consensus operator are the one that offer a more exact and stable modeling \cite{BenAlaya:2016}.

\subsection{Implementations of graph filters}
\label{ssec:chapfilt:ImplementationGFilt}

Given a frequency response $h$, how to effectively implement the action of the corresponding graph filter  $\mathbf{H}$ on a graph signal $\mathbf{x}$? \\

\noindent\textbf{The direct approach.} It consists in first diagonalizing $\mathbf{R}$ to obtain $\mathbf{U}$ and $\mathbf{\Lambda}$; then computing the filter matrix $\mathbf{H}=\mathbf{U}h(\mathbf{\Lambda})\mathbf{U}^\top$; before finally left-multiplying the graph signal $\mathbf{x}$ by $\mathbf{H}$ to obtain its filtered version. The overall computational cost of this procedure is $\mathcal{O}(N^3)$ arithmetic operations due to the diagonalization and $\mathcal{O}(N^2)$ memory space as the graph Fourier transform $\mathbf{U}$ is in general dense.\\


\noindent\textbf{The polynomial approximate filtering approach.} More efficiently though, we can first, quickly estimate $\lambda_\text{min}$ and $\lambda_\text{max}$ (for instance via the power method) and second, look for a polynomial that best approximates $h(\lambda)$ on the whole interval $[\lambda_\text{min},\lambda_\text{max}]$. Let us call $\tilde{h}_i$ the coefficients of this approximate polynomial. We have:
\begin{align}
 \mathbf{Hx} = \mathbf{U}h(\mathbf{\Lambda})\mathbf{U}^\top\mathbf{x}\simeq \mathbf{U}\sum_{i=0}^{p} \tilde{h}_i \mathbf{\Lambda}^i\mathbf{U}^\top \mathbf{x}=\sum_{i=0}^{p} \tilde{h}_i \mathbf{R}^i \mathbf{x}.
\end{align}
The number of required arithmetic operations is $\mathcal{O}(p|\mathcal{E}|)$, where $p$ is a trade-off between precision and computational cost. Also, the easier is $h$ approachable by a low-order polynomial, the better the approximation. At the same time, the larger $p$ is, the more accurate the approximation is.  The authors of~\cite{shuman_DCOSS2011} recommend as a natural choice for the approximation polynomials, the Chebychev polynomials as they are known to be optimal in the $\infty$-norm sense.
In some circumstances, other choices can be preferred. For instance, when approximating the ideal low-pass,  Chebychev polynomials yield Gibbs oscillations around the cut-off frequency that turn penalising  for smooth filters. In that case, other choices are possible~\cite{tremblay_compressive_2016}, such as the Jackson-Chebychev polynomials that attenuate such unwanted oscillations. \\

\noindent\textbf{The Lanczos approximate filtering approach.} In~\cite{susnjara_accelerated_2015}, and based on works by~\cite{gallopoulos_efficient_1992}, authors propose an approximate filtering approach based on Lanczos iterations. Given a signal $\mathbf{x}$, the Lanczos algorithm computes an orthonormal basis $\mathbf{V}_p\in\mathbb{R}^{N\times p}$ of the Krylov subspace associated to $\mathbf{x}$: $K_p(\mathbf{\mathbf{L}},\mathbf{x})=\text{span}(\mathbf{x},\mathbf{\mathbf{L}x},\ldots,\mathbf{\mathbf{L}}^{p-1}\mathbf{x})$, as well as a small tridiagonal matrix $\mathbf{H}_p\in\mathbb{R}^{p\times p}$ such that: $\mathbf{V}_p^*\mathbf{\mathbf{L}}\mathbf{V}_p=\mathbf{H}_p$. The approximate filtering then reads:
\begin{align}
 \mathbf{Hx}\simeq ||\mathbf{x}||_2 \mathbf{V}_ph(\mathbf{H}_p)\bm{\delta}_1.
\end{align}
At fixed $p$, this approach has a typical complexity in $\mathcal{O}(p|\mathcal{E}|)$, possibly raised to $\mathcal{O}(p|\mathcal{E}|+Np^2)$ if a reorthonomalisation is needed to stabilize the algorithm; 
a cost that is comparable to that of the polynomial approximation approach. 
Theoretically, the quality of approximation is similar (see~\cite{susnjara_accelerated_2015} for details). In practice however, it has been observed that if the spectrum is regularly spaced, polynomial approximations should be prefered, while the Lanczos method has an edge over others in the case of irregularly-spaced spectra. This is understandable as Krylov subspaces are also used for diagonalisation purposes (see for instance, chapter 6 of~\cite{saad_numerical_2011}) and thus naturally adapt to the underlying spectrum.\\


\begin{figure}\centering
 \includegraphics[width=\textwidth]{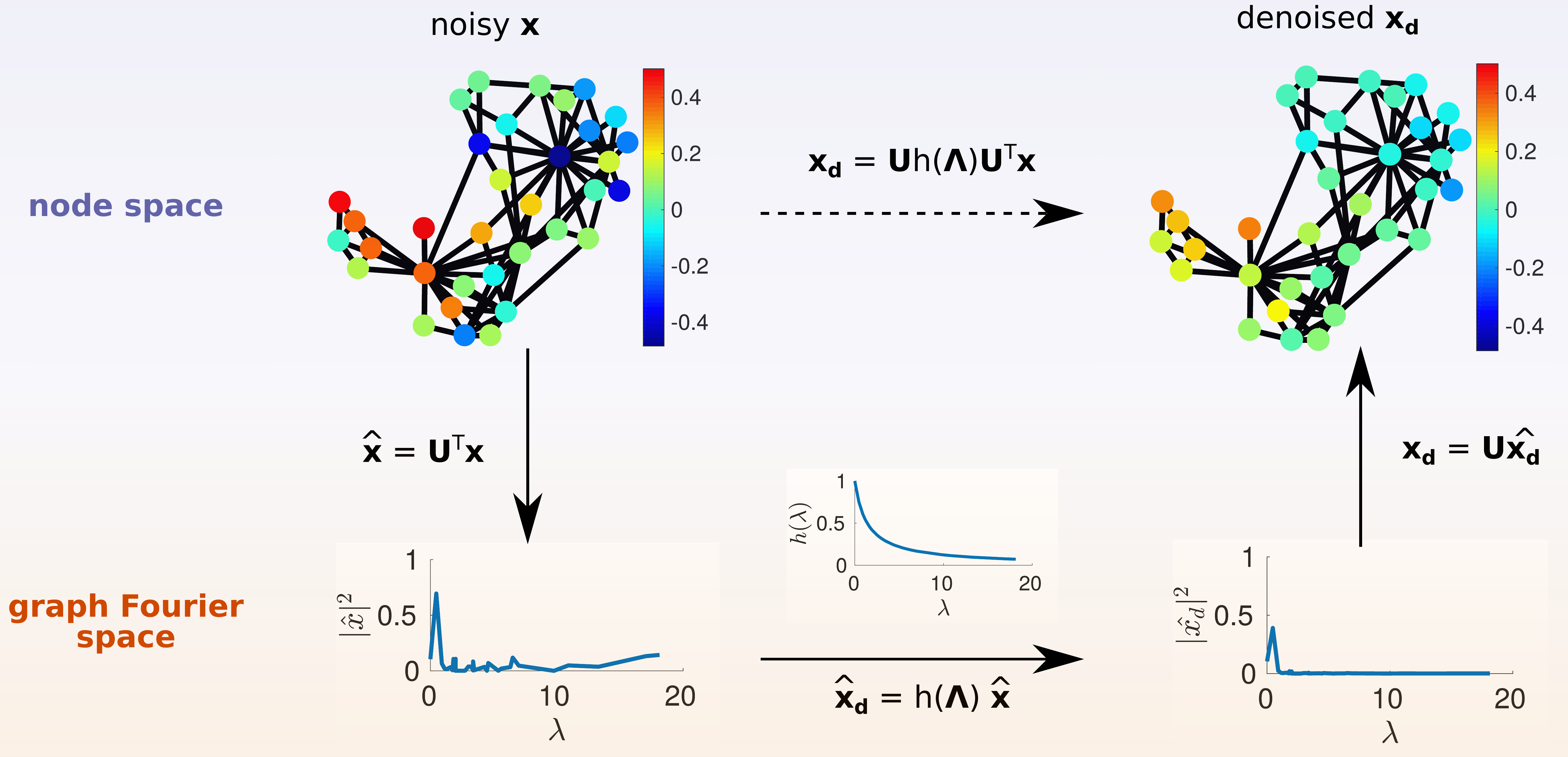}
 \caption{\label{chapfilt:fig:denoising} \textbf{Illustration of graph filters: a denoising toy experiment.} The  input signal $\mathbf{x}$ is a noisy version (additive Gaussian noise) of the low-frequency graph signal displayed in Fig.~\ref{chapfilt:fig:illus}. We show here the filtering operation in the graph Fourier domain associated to $\mathbf{R}=\mathbf{L_n}$. }
\end{figure}

\noindent \textbf{Distributed implementation of ARMA filters.}
The ARMA filters being defined through a rational fraction, are IIR filters. Henceforth, the polynomials approach to distribute and fasten the computation are not the most efficient ones. The methods developed in \cite{loukas_ARMA_2015,shi_infinite_2015} yield a distributed implementation of ARMA filters. The first point is to remember that the rational filter off eq.~(\ref{chapfilt:eq:ARMAfilter})
can be implemented from its partial fraction decomposition as a sum of polynomial fractions of order 1 only. Then, the distributed implementation of filtering $\mathbf{x}$ can be done by studying the 1st order recursion~\cite{loukas_ARMA_2015}:
\begin{equation}
\label{chapfilt:eq:ARMA1recursion}
\textbf{y}(t+1) = c \textbf{M}  \textbf{y}(t) + d \mathbf{x}
\end{equation}
where $\mathbf{y}(t+1)$ is the filter output at iteration $t$.
The operator $\textbf{M} $ is chosen equal to $(\lambda_{\text{max}} - \lambda_{\text{min}}) \mathbf{I} - \mathbf{R}$, with the same eigenvectors as $\mathbf{R}$, and with a minimal spectral radius that ensures good convergence properties for the recursion. 
The coefficient $c$ and $d$ are chosen in $\mathbb{C}$ so that, with $r = -d/c$ and $\rho=1/c$,
the proposed recursion reproduces the effect of the following ARMA(1,0) filter:
$ h(\lambda) = r / (\lambda - \rho)$. The coefficients $r$ and $\rho$ are the residue and the pole of the rational function, respectively. 

Because $\textbf{M}$ is local in the graph, the recursive application of eq.~(\ref{chapfilt:eq:ARMA1recursion}) is local and the algorithm is then naturally distributed on the graph, with a memory and operation at each recursion in $\mathcal{O}( K|\mathcal{E}|)$ for $K$ filters in parallel to compute the output of a ARMA$(K,K)$ filter.
This approach is shown in \cite{loukas_ARMA_2015,shi_infinite_2015} to converge efficiently. Moreover, in \cite{isufi_autoregressive_2017}, the behavior of this design is also studied in time-varying settings, when the graph and the signal are possibly time-varying; it is then shown that the recursion can remain stable and usable as a distributed implementation of IIR filters.
\\

\noindent\textbf{Illustration.} We show in Fig.~\ref{chapfilt:fig:denoising} an example of a filtering operation on a graph signal. We consider here the case of a Tikhonov denoising (see Eq.~\eqref{eq:chapfilt:tikhonov}), i.e., with a frequency response equal to $h(\lambda)=1/(1+\gamma\lambda)$. 

	\section{Filterbanks and multiscale transforms on graphs}
\label{chapfilt:sec3}

%
%



To process and filter signals on graphs, it is attractive to develop equivalent of multiscale
transforms such as wavelets, i.e. ways to decompose a graph signal on components at different scales,
or frequency ranges. 
The road to multiscale transforms on graphs has originally been tackled in the vertex domain~\cite{crovella_graph_2003} to analyse data on networks.
Thereafter, a general design of multiscale transforms was based on the diffusion of
signals on the graph structure, leading to the powerful framework of Diffusion Wavelets \cite{coifman_ACHA2006,gavish_multiscale_2010}. This latter works were already based on a diffusion
operators, usually a Laplacian or a random walk operator, whose powers are decomposed in order to obtain a multiscale orthonormal basis. The objective was to build a kind of equivalent of discrete
wavelet transforms for graph signals. 

In this chapter, we focus on two other constructions of multiscale transforms on graphs which are
more related to the Graph Fourier Transform. The frequency analysis and filters described 
here are:
1) the method of~\cite{hammond_ACHA2011} which develops an analog of continuous wavelet transform on graphs,
2) approaches that combine filters on graphs with graph decompositions through 
decimation (pioneered in \cite{narang_TSP2012} with decimation of bipartite graphs) 
or aggregation of nodes; in a nutshell, these methods are very close to the filter banks implementation of 
discrete wavelets \cite{strang_book1996}.

\subsection{Continuous multi-scale transforms}

In the following, we work with undirected graphs and $\mathbf{R}=\mathbf{L_n}=\mathbf{I}-\mathbf{D}^{-\frac{1}{2}}\mathbf{A}\mathbf{D}^{-\frac{1}{2}}$, whose eigenvalues are contained in the interval $[0, \, 2]$. The generalization to other operators can be done using the guidelines of previous sections.

The first continuous multiscale transform based on the graph Fourier transform was introduced via the spectral graph wavelet transform~\cite{hammond_ACHA2011} (and inherent of some properties of the concept of diffusion polynomial frames~\cite{maggioni_diffusion_2008}). These wavelets were defined by analogy to the classical wavelets in the following sense. Classically, a wavelet family $\{\psi_{s,\tau}(t)\}$ centered around time $\tau$ and at scale $s$ is the translated and dilated version of a mother wavelet $\psi(t)$, generally defined as a zero-mean, square integrable function. Mathematically it is expressed for all $s\in\mathbb{R}^+$ and $\tau\in\mathbb{R}$ as:
\begin{align}
\psi_{s,\tau}(t)=\frac{1}{s} \psi\left(\frac{t-\tau}{s}\right)\; 
\end{align}
or equivalently, in the frequency domain, with $\mathbb{F}$ the continuous Fourier transform:
\begin{align}
\label{chapfilt:eq:wavelet_in_Fourier}
\mathbb{F}[\psi_{s,\tau}](\omega) = \mathbb{F}[\psi](s\omega) \; \mathbb{F}[\delta_\tau](\omega) = \widehat{\psi}(s\omega) \; e^{-i\omega \tau}.
\end{align}
Then, for a signal $x$, the wavelet coefficient at scale $s$ and instant $\tau$ is given by the inner product
$\mathbb{W}_{s,\tau}x = \langle x, \psi_{s,\tau} \rangle$.

By analogy, transposing eq.~(\ref{chapfilt:eq:wavelet_in_Fourier}) with GFT, a spectral graph wavelet $\psi_{s,a}$ at scale $s$ and node $a$ reads\footnote{Note that the scale parameter stays continuous, but the localization parameter is discretized to the set of nodes $a$ of the graph.}:
\begin{align}
 \psi_{s,a} = \mathbf{U} \, h(s\Lambda) \, \mathbf{U}^\top \bm{\delta}_a,
\end{align}
where the graph filter $h(\lambda)$ plays the role of the wavelet bandpass filter $\widehat{\psi}(\omega)$. 
The shifted scaled wavelet identifies to the impulse response of $h(s\lambda)$ to a Dirac localized on node $a$.
In particular, the shape of the filter originally proposed in~\cite{hammond_ACHA2011} was: 
\begin{equation}
  h^{\textsc{SGW}}(\lambda)= 
  \left\{
  \begin{aligned}
        &\lambda_{*}^{-\alpha} \lambda^\alpha~\mbox{ for }  \lambda<\lambda_{1}\\
        &q(\lambda) ~\mbox{ for }  \lambda_{1}\leq \lambda\leq\lambda_{2}\\
        &\lambda_{2}^{\beta} \lambda^{-\beta}~\mbox{ for }  \lambda>\lambda_{2},
        \end{aligned}
    \right.\nonumber
\end{equation}
with $\alpha$, $\beta$, $\lambda_{1}$ and $\lambda_{2}$ four parameters, and $q(\lambda)$ the unique cubic polynomial interpolation that preserves continuity and the derivative's continuity. Several properties on the obtained wavelets may be theoretically derived, for instance the notion of locality (the fact that wavelets' energy is mostly contained around the node on which it is centered). Given a selection of scales $\mathcal{S}=(s_1,\ldots,s_m)$ and a graph signal $\mathbf{x}$, the signal wavelet coefficient associated to the node $a$ and the scale $s\in\mathcal{S}$ reads: $\mathbb{W}_{s,a}\mathbf{x}=\psi_{s,a}^\top\mathbf{x}$. Then, a question that naturally arises is that of invertibility of the wavelet transform: can one recover any signal $\mathbf{x}$ from its wavelet coefficients? As defined here, the wavelet transform is not invertible as it does not take into account --due to the zero-mean constraint of the wavelets-- the signal's information associated to the null frequency, i.e., associated to the first eigenvector $\mathbf{u}_1$. To enable invertibility, one may simply add any low-pass filter $h_0(\lambda)$ to the set of filters $\{h^{\textsc{SGW}}(s\lambda)\}_{s\in\mathcal{S}}$ of the wavelet transform. We write $\phi_{a}$ their associated atoms:
\begin{align}
 \phi_{a} = \mathbf{U}h_0(\mathbf{\Lambda})\mathbf{U}^\top\mathbf{\delta}_a.
\end{align}
The following theorem derives:
\begin{theorem}[Theorem 5.6 in~\cite{hammond_ACHA2011}] 
 Given a set of scales $\mathcal{S}$, the set of atoms $\{\{\psi_{s,a}\}_{s\in\mathcal{S}}\cup \{\phi_a\}\}_{a\in\mathcal{V}}$ forms a frame with bounds $A$, $B$ given by:
 \begin{align}
  A &= \min_{\lambda\in[0, \lambda_{\text{max}}]} G (\lambda)\\
  B &= \max_{\lambda\in[0, \lambda_{\text{max}}]} G (\lambda)
 \end{align}
where $G(\lambda) = (h_0(\lambda))^2 + \sum_{s\in\mathcal{S}} \; (h^{\textsc{SGW}}(s\lambda))^2$.
\end{theorem}
In theory, invertibility is guaranteed provided that $A$ is different from 0. Nevertheless, in practice, one should strive to design filter shapes (wavelet and low-pass filters) and to choose a set of scales such that $A$ is as close as possible to $B$ in order to deal with well-conditioned inverses. Doing so, we obtain the so-called tight (or snug) frames, i.e. frames such that $A=B$ (or $A\approx B$). An approach is to use classical dyadic decompositions using bandlimited filters such as in Table 1 of~\cite{leonardi_TSP2013}.  Another desirable property of such frames is their discriminatory power: the ability at discerning different signals only by considering their wavelet coefficients. For a filterbank to be discriminative, each filter element needs to take into account information from a similar number of eigenvalues of the Laplacian. The eigenvalues of an arbitrary graph being unevenly spaced on $[0,\lambda_{\text{max}}]$, one needs to compute or estimate the exact density of the spectrum of the graph under consideration~\cite{shuman_TSP2015}.

\subsection{Discrete multi-scale transforms}

A second general way to generate multiscale transforms is via a succession of filtering and decimation operations, as in Fig.~\ref{chapfilt:fig:FBscheme}.  This scheme is usually cascaded as in~\ref{chapfilt:fig:FBscheme_cascade}, and each level of the cascade represents a scale of description of the input signal. Thereby, as soon as decimation enters into the process, we talk about ``discrete'' multi-scale transforms, as the scale parameter can no longer be continuously varied. For details on this particular approach to multiscale transforms in classical signal processing, we refer e.g. to the book by Strang and Nguyen~\cite{strang_book1996}. In the following, we directly consider the graph-based context. Let us first settle notations:
\begin{itemize}
 \item The decimation operator may be generally defined by  partitioning the set of nodes $\mathcal{V}$ into two sets $\mathcal{V}_0$ and $\mathcal{V}_1$. As this subdivision is a partition, we have $\mathcal{V}_0 \cup \mathcal{V}_1 = \mathcal{V}$ and $\mathcal{V}_0 \cap \mathcal{V}_1=\emptyset$. Moreover, let us define $\downarrow_{\mathcal{V}_i}$ the downsampling operator associated to $\mathcal{V}_i$: given any graph signal $\mathbf{x}$, $\mathbf{y}_i=\downarrow_{\mathcal{V}_i}\mathbf{x}$ is the reduction of $\mathbf{x}$ to $\mathcal{V}_i$. We also define the upsampling operator $\uparrow_{\mathcal{V}_i}=\downarrow_{\mathcal{V}_i}^\top$. Given $\mathbf{y}_i$ a signal defined on $\mathcal{V}_i$, $\uparrow_{\mathcal{V}_i}\mathbf{y}_i$ is the zero-padded version of $\mathbf{y}_i$ on the whole graph. The combination of both operators reads: $\uparrow_{\mathcal{V}_i}\downarrow_{\mathcal{V}_i}=\text{diag}(\mathcal{I}_{\mathcal{V}_i})$, where $\mathcal{I}_{\mathcal{V}_i}$ is the indicator function of $\mathcal{V}_i$. Moreover, we define
 \begin{align}
  \mathbf{J} = 2\uparrow_{\mathcal{V}_0}\downarrow_{\mathcal{V}_0} - \; \mathbf{I} = \text{diag}(\mathcal{I}_{\mathcal{V}_0})-\text{diag}(\mathcal{I}_{\mathcal{V}_1}).
 \end{align}
 \item We define $\mathbf{H}_0$ and $\mathbf{H}_1$ respectively a low-pass and a high-pass graph filter, called analysis filters. $\mathbf{G}_0$ and $\mathbf{G}_1$ are also two graph filters, called synthesis filters. All filters are associated to their frequency responses $h_0(\lambda)$, $h_1(\lambda)$, $g_0(\lambda)$ and $g_1(\lambda)$.
\end{itemize}
The signal $\mathbf{y}_0=\downarrow_{\mathcal{V}_0}\mathbf{H}_0\mathbf{x}$ is called the approximation of $\mathbf{x}$, whereas $\mathbf{y}_1=\downarrow_{\mathcal{V}_1}\mathbf{H}_1\mathbf{x}$ is generally understood as the necessary details to recover $\mathbf{x}$ from its approximation.

Given the scheme of Fig.~\ref{chapfilt:fig:FBscheme}, one writes the processed signal $\tilde{\mathbf{x}}$ as:
\begin{align}
 \tilde{\mathbf{x}} &= \left(\mathbf{G}_0\uparrow_{\mathcal{V}_0}\downarrow_{\mathcal{V}_0}\mathbf{H}_0 + \mathbf{G}_1\uparrow_{\mathcal{V}_1}\downarrow_{\mathcal{V}_1}\mathbf{H}_1\right)\mathbf{x}\\
 &= \frac{1}{2}\left(\mathbf{G}_0\mathbf{H}_0+\mathbf{G}_1\mathbf{H}_1\right)\mathbf{x} + \frac{1}{2}\left(\mathbf{G}_0\mathbf{J}\mathbf{H}_0 - \mathbf{G}_1\mathbf{J}\mathbf{H}_1 \right)\mathbf{x}.
\end{align}
When designing such discrete filterbanks, and in order to enable perfect reconstruction ($\forall\mathbf{x}\in\mathbb{R}^N~~\tilde{\mathbf{x}}=\mathbf{x}$), one deals with two main equations linking all four filters and the matrix $\mathbf{J}$:
\begin{align}
 \mathbf{G}_0\mathbf{H}_0+\mathbf{G}_1\mathbf{H}_1 = 2\mathbf{I}\\
 \mathbf{G}_0\mathbf{J}\mathbf{H}_0 - \mathbf{G}_1\mathbf{J}\mathbf{H}_1 = \mathbf{0}
\end{align}
Left-multiplying by $\mathbf{U}^\top$ and right-multiplying by $\mathbf{U}$, one obtains equivalently:
\begin{align}
\label{chapfilt:eq:dfb_bp_1} g_0(\mathbf{\Lambda})h_0(\mathbf{\Lambda})+g_1(\mathbf{\Lambda})h_1(\mathbf{\Lambda}) = 2\mathbf{I}\\
\label{chapfilt:eq:dfb_bp_2} g_0(\mathbf{\Lambda})\mathbf{U}^\top\mathbf{J}\mathbf{U}h_0(\mathbf{\Lambda}) - g_1(\mathbf{\Lambda})\mathbf{U}^\top\mathbf{J}\mathbf{U}h_1(\mathbf{\Lambda}) = \mathbf{0}
\end{align}
Eq.\eqref{chapfilt:eq:dfb_bp_1} is purely spectral, and may be seen as a set of $N$ equations: 
\begin{align}
\label{chapfilt:eq:dfb_bp_1_spec}
 \forall\lambda_i~~~g_0(\lambda_i)h_0(\lambda_i)+g_1(\lambda_i)h_1(\lambda_i)=2.
\end{align}
On the other hand, Eq.~\eqref{chapfilt:eq:dfb_bp_2} is not so simple due to the decimation operation and needs to be investigated in detail. 

\begin{figure}\centering
 \includegraphics[width=0.7\textwidth]{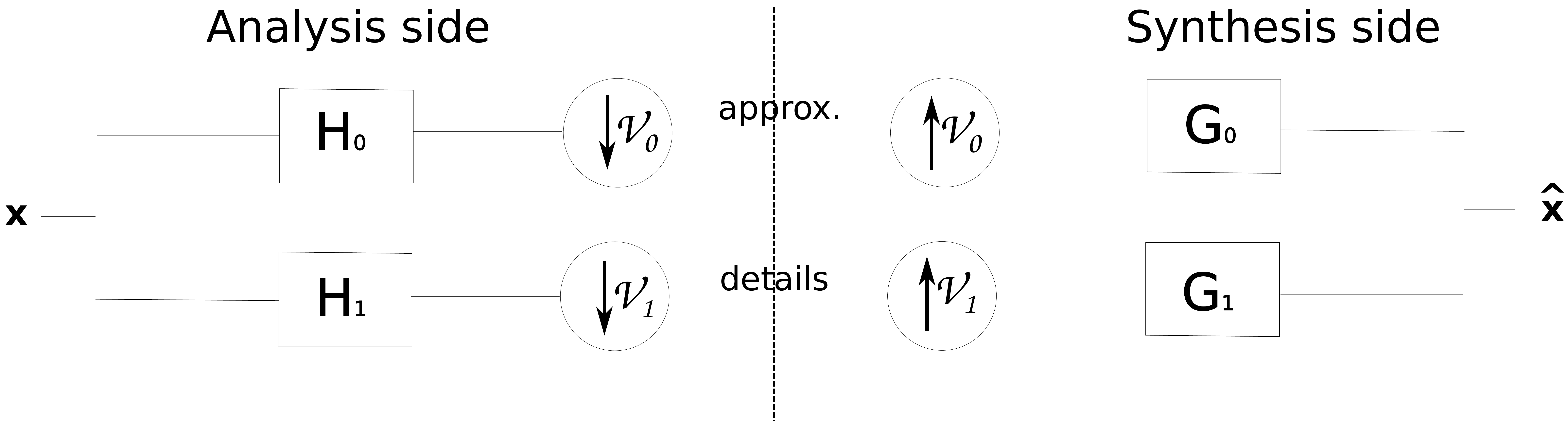}
 \title{}\caption{\label{chapfilt:fig:FBscheme}\textbf{A filterbank seen as a succession of filtering and decimating operators.}}
\end{figure}

\begin{figure}\centering
 \includegraphics[width=\textwidth]{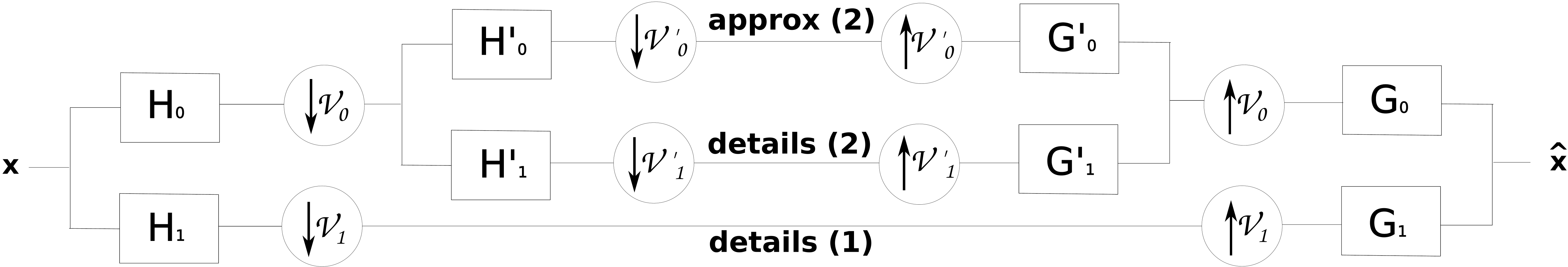}
 \title{}\caption{\label{chapfilt:fig:FBscheme_cascade}\textbf{A cascaded filterbank (here two levels).}}
\end{figure}

In 1D classical signal processing (equivalent to the undirected circle graph), the decimation operator samples one every two nodes. Moreover, given $\mathbf{x}'=
\uparrow_{2}\downarrow_{2}\mathbf{x}$, one classically has the following aliasing phenomenon (Theorem 3.3 of~\cite{strang_book1996}):
\begin{align}
 \mathbb{F}[\mathbf{x}'](\omega)=\frac{1}{2}\left(\mathbb{F}[\mathbf{x}](\omega)+\mathbb{F}[\mathbf{x}](\omega+\pi)\right).
\end{align}
This means that the decimation operations may be explictly described in the Fourier space, which greatly simplifies calculations by enabling to write Eq.~\eqref{chapfilt:eq:dfb_bp_2} as a purely spectral equation as well. Moreover, this also  entails that the combined filtering-decimation operations may be understood as a global multiscale filter (yet with discrete scales), thereby connecting with the previous approach.  
Now, the central question that remains is how to mimick the filtering/decimation approach on general graphs?

In Section~\ref{chapfilt:subsec:struc}, it will be shown that for very well structured graphs such as bipartite graphs or $m$-cyclic graphs, generalizations of the decimation operator may be defined and their effect explicitly described as graph filters. 
Then the case of arbitrary graph is studied in~\ref{chapfilt:subsec:arbitrary} where the ``one every two nodes'' paradigm is not transposable directly;  several approaches to do that will then be reviewed. 

\subsubsection{Filterbanks on bipartite graphs and other strongly structured graphs}
\label{chapfilt:subsec:struc}

\noindent\textbf{Filterbanks on bipartite graphs.} 
Bipartite graphs are graphs where the nodes are partitioned in two sets of nodes $\mathcal{A}$ and $\mathcal{B}$ such that all links of the graph connect a node in $\mathcal{A}$ with a node in $\mathcal{B}$. On bipartite graphs, the ``one-every-two-node'' paradigm has a natural extension: decimation ensembles are set to $\mathcal{V}_0=\mathcal{A}$ and $\mathcal{V}_1=\mathcal{B}$. Leveraging the fact that bipartite graphs' spectra are symmetrical\footnote{i.e., if $\lambda$ is an eigenvalue of $\mathbf{\mathcal{L}}$, then so is $2-\lambda$} around the value 1, Narang and Ortega~\cite{narang_TSP2012} show the bipartite graph spectral folding phenomenon:
\begin{align}
 \forall \lambda \quad\text{\textbf{Pr}}_\lambda \mathbf{J} = \mathbf{J} \text{\textbf{Pr}}_{2-\lambda}, 
\end{align}
(with $\textbf{Pr}_\lambda$ as in eq. \ref{chapfilt:eq:R_as_projectors}).
This means that for any  filter $\mathbf{H}=\mathbf{U}h(\mathbf{\Lambda})\mathbf{U}^\top$ one has:
\begin{align}
 \mathbf{GJ} = \mathbf{U}g(\mathbf{\Lambda})\mathbf{U}^\top\mathbf{J} = \mathbf{J} \mathbf{U}g(2\mathbf{I}-\mathbf{\Lambda})\mathbf{U}^\top.
\end{align}
Eq.~\eqref{chapfilt:eq:dfb_bp_2} therefore boils down to a second set of $N$ purely spectral equations:
\begin{align}
\label{chapfilt:eq:dfb_bp_2_spec}
 \forall \lambda_i\quad g_0(2-\lambda_i)h_0(\lambda_i)-g_1(2-\lambda_i)h_1(\lambda_i)=0.
\end{align}
Eqs.~(\ref{chapfilt:eq:dfb_bp_1_spec}) and~(\ref{chapfilt:eq:dfb_bp_2_spec}) give us $2N$ equations linking the $4N$ parameters of the four filters to ensure perfect reconstruction. The other $2N$ degrees of liberty are free to be used to design filterbanks with other desirable properties of filterbanks, such as (bi-)orthogonality, compact-supportness of the atoms, and of course also to adapt to the specific application for which these filters are designed. \\

\noindent\textbf{Filterbanks on other regular structures.} Extending these ideas, several authors have proposed similar approaches to define filterbanks on other regular structures such as $M$-block cyclic graphs~\cite{teke_extending_2017} or circulant graphs~\cite{ekambaram_circulant_2013,Kotzagiannidis_2016}. In any case, writing decimation operations exactly as graph filters requires regular structures on graphs inducing at least some regularity in the spectrum one may take advantage of. 
All these approaches lead to exact reconstruction procedures.
However, arbitrary graphs do not have such regularities, and other approaches are required.

\subsubsection{Filterbanks on arbitrary graphs}
\label{chapfilt:subsec:arbitrary}

To extend the filterbanks approach to arbitrary graphs, one needs to either generalize the decimation operators, or
to bypass it by working with operator of aggregation of nodes.
We discuss in this last part some solutions that were proposed in the literature. 
A complementary and thoughtful discussion on graph decimation, graph aggregation and graph reconstruction can be found in sections III and IV of ~\cite{shuman_ARXIV2013}.
For that, the two key questions are:
\begin{itemize}
 \item how to generalize the decimation operator on arbitrary graphs? We will see that generalized decimation operators either try to mimick the classical decimation and attempt to sample ``one every two nodes'', or aggregate nodes to form supernodes according in general to some graph cut objective function.
 \item how to build the new coarser-scale graph from the decimated nodes (or aggregated supernodes)? In fact, after each decimation, if one wants to cascade the filterbank, a new coarse-grain graph has to be built in order to define the next level's graph filters. The nodes (resp. supernodes) are set thanks to decimation (resp. aggregation): but how do we link them together? 
\end{itemize}

\noindent\textbf{Graph decimation.} The first work to generalize filterbanks on arbitrary graphs is due to Narang and Ortega~\cite{narang_TSP2012} and consists in decomposing the graph into an edge-disjoint collection of bipartite subgraphs, and then to apply the scheme presented in Section~\ref{chapfilt:subsec:struc} on each of the subgraphs. In this collection, each subgraph has the same node set, and the union of all subgraphs sums to the original graph. 
To perform this decomposition (which is not unique), the same authors propose a coloring-based algorithm, called Harary's decomposition. Sakiyama and Tanaka~\cite{sakiyama_TSP2014} also used this decomposition as one of their design's cornerstone. Unsatisfied by the NP-completeness of the coloring problem (even though heuristics exist), Nguyen and Do~\cite{nguyen_TSP2015} propose another decomposition method based on maximum spanning trees. 

The bipartite paradigm's main advantage comes from the fact that decimation has an explicit formulation in the graph's Fourier space, thereby enabling exact filter designs depending on the given task. In our opinion, when applied to arbitrary graphs, its main drawback comes from the non-unicity of the bipartite subgraphs decomposition, as well as the seamingly arbitrariness of such a decomposition: from a graph signal point-of-view, what is the meaning of a bipartite decomposition? Letting go of this paradigm, and slightly changing the general filterbank design presented in Section~\ref{chapfilt:subsec:struc}, other generalized graph decimations have been proposed. For instance, in~\cite{narang_ICIP2010}, authors propose to separate the graph in two sets $\mathcal{V}_0, \mathcal{V}_1$ according to its max cut, i.e., maximizing $\sum_{i\in\mathcal{V}_0}\sum_{j\in\mathcal{V}_1}\mathbf{W}_{ij}$. In~\cite{shuman_ARXIV2013}, authors suggest similarly to partition the graph into two sets according to the polarity of the last eigenvector (i.e. the eigenvector associated to the highest frequency). In~\cite{avena_intertwining_2017}, authors use an original approach based on random forests to sample nodes, where they have a probabilistic version of ``equally-spaced'' nodes on the graph. \\ 


\noindent\textbf{Graph aggregation.} Another paradigm in graph reduction is graph aggregation, where, instead of \emph{selecting} nodes as in decimation, \emph{aggregates} entire regions of the graph in ``supernodes''. In general, these methods are based on first clustering the nodes in a partition $\mathcal{P}=\{\mathcal{V}_1, \mathcal{V}_2, \ldots, \mathcal{V}_J\}$. Each of these subsets will define a supernode of the coarse graph: this reduced graph thus contains $J$ supernodes. Once a rule is chosen to connect these supernodes together (the object of the next paragraph), the coarse graph is fully defined, and the method may be iterated to obtain a multiresolution of the initial graph's structure. All these methods differ mainly on the choice of the algorithm or the objective function to find this partition. For instance, one may find methods based on random walks~\cite{lafon_diffusion_2006}, on short time diffusion distances~\cite{livne_lean_2012}, on the algebraic distance~\cite{ron_relaxation-based_2011}, etc. Other multiresolution approaches may also be found in~\cite{dhillon_weighted_2007, karypis_fast_1998,murtagh_algorithms_2012}. One may also find many approaches from the network science community in the connex field of community detection~\cite{newman_modularity_2006, fortunato_community_2010}, and in particular multiscale community detection~\cite{reichardt_statistical_2006, schaub_markov_2012, tremblay_graph_2014-1}. All these methods are concerned about providing a multiresolution description of the graph structure, but do not consider any graph signal. Recently, graph signal processing filterbanks have been proposed to define a multiscale representation of graph signals based on these approaches. In~\cite{tremblay_subgraph-based_2016}, we proposed such an approach where we define a generalized Haar filterbank: instead of averaging and differentiating over pairs of nodes as in the classical Haar filterbank, we average and differentiate over the subsets $\mathcal{V}_j$ of a partition in subgraphs. In~\cite{irion_applied_2015}, authors propose a similar approach and define other types of filterbanks such as the hierarchical graph Laplacian eigentransform. Another similar Haar filterbank may be found in~\cite{gavish_multiscale_2010}. All these methods are independent of exactly which aggregation algorithm one chooses to find the partition. Let us also cite methods that provide multiresolution approaches without necessarily defining low-pass and high-pass filters explictly in the graph Fourier domain~\cite{lee_treelets_2008,gavish_multiscale_2010,mishne_data_driven_2017}. Finally, let us point out that these methods may be extended to graph partitions $\mathcal{P}$ containing overlaps, as in~\cite{szlam_diffusion-driven_2005}. \\

\noindent\textbf{Coarse graph reconstruction.} Once one decided how to decimate nodes, or how to partition them in supernodes, how should one connect these nodes together in order to form a consistent reduced graph?  
In order to satisfy constraints such as interlacement (coarsely speaking, that the spectrum of the reduced graph is representative of the spectrum of the initial graph) and sparsity, Shuman et al.~\cite{shuman_ARXIV2013} propose a Kron reduction followed by a sparsification step. The Laplacian of the reduced graph is thus defined as the Schur complement of the initial graph's Laplacian relative to the unsampled nodes. The sparsification step is performed via a sparsifier based on effective resistances by Spielman and Srivastava~\cite{spielman_graph_2011} that approximately preserves the spectrum. In~\cite{avena_intertwining_2017}, authors propose another approach to the intuitive idea that the initial and coarse graph should have similar spectral properties: they look for the coarse Laplacian matrix that satisfies an intertwining relation. In~\cite{nguyen_TSP2015}, authors connect nodes according to the set of nested bipartite graphs obtained by their maximum spanning tree algorithm. In aggregation methods~\cite{tremblay_subgraph-based_2016, irion_applied_2015}, there is an inherent natural way of connecting supernodes: the weight of the link between supernodes $i$ and $j$ is equal to the sum of the weights of the links connecting nodes in $\mathcal{V}_i$ to nodes in $\mathcal{V}_j$.   
\\

\noindent \textbf{Illustrations.} We show in Fig.~\ref{chapfilt:fig:FBexample} an example of multiresolution analysis of a given graph signal on the Minnesota traffic graph, using a method of successive graph aggregation to compute the details and approximations at different scales. The specific method for this illustration is the one detailed in~\cite{tremblay_subgraph-based_2016}.

\begin{figure}\centering
 \includegraphics[width=\textwidth]{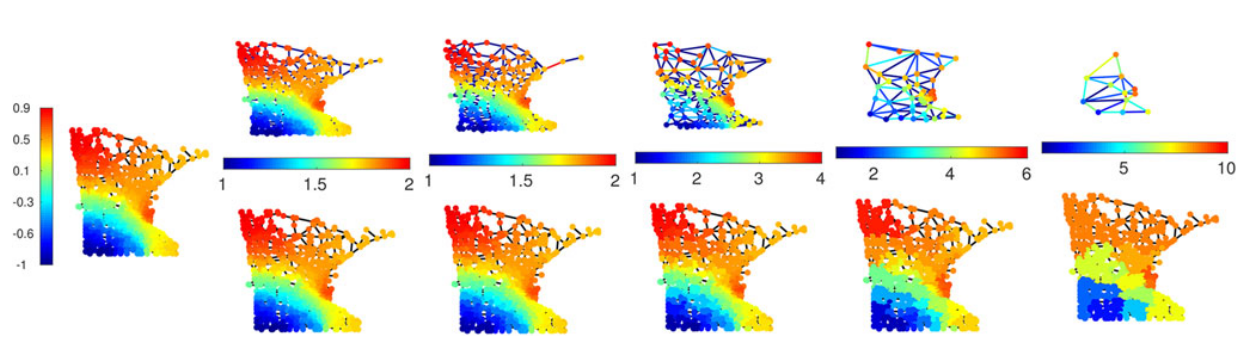}
 \caption{\label{chapfilt:fig:FBexample} \textbf{An example of multiresolution analyis of a graph signal (from~\cite{tremblay_subgraph-based_2016}).} Left: original smooth graph signal (sum of the five lowest Fourier modes normalized by its maximum absolute value) defined on the Minnesota 
traffic graph. The vertical colorbar of this figure is valid for all graph signals represented on this figure. Top row: successive approximations of the graph signal. The horizontal colorbar on the bottom of each figure corresponds to the weights of the links of the corresponding coarsened graph. 
Figures who do not have a bottom horizontal colorbar represent binary graphs. Lower row: for each of the successive approximations, we represent the upsampled reconstructed graph signal obtained from the corresponding approximations.}
\end{figure}

\section*{Conclusion}

The purpose of this chapter was to introduce the reader to a basic understanding of what is a Graph Fourier Transform. We stressed how it can be generally introduced for undirected or directed graphs, by choosing a reference operator whose spectral domain will define the frequency domain for graph signals. Then, we led the reader to the more elaborated designs of graph filters and multiscale transforms on graphs.
The first section is voluntarily introductory, and almost self-contained. 
Indeed, we endeavoured to delineate a general guideline that shows, in an original manner, that there is no major discrepancy between choosing a Laplacian, an Adjacency, or a random walk operator\ldots as long as one chooses accordingly, the appropriate notion of frequency to analyse graph signals.  
Then, after a proper definition of graph filters, our objective was to review the literature and to propose guidelines and pointers to the relevant results on graph filters and related multiscale transforms. This last part being written as a review, we beg for reader's indulgence, as many details are skipped, and some works are only reported here in a sketchy manner.

	\section*{Acknowledgements}
	\ANR. 
\else
	\maketitle
	
	\tableofcontents

	\section*{Acknowledgements}
	\ANR
\fi

\ifels
	\bibliographystyle{elsarticle-num}
	\section*{References}
\else
	\bibliographystyle{IEEEtran}
\fi
\bibliography{biblio_GraphFilters}

\end{document}